\documentclass[12pt]{article}
\usepackage[top=1.25in, bottom=1.25in, left=1.25in, right=1.25in]{geometry} 
\usepackage[utf8]{inputenc} 
\usepackage{amssymb}
\usepackage{amsmath} 
\usepackage{mathrsfs}
\usepackage[T1]{fontenc}
\usepackage{newtxtext,newtxmath} 
\usepackage{appendix}
\usepackage{ntheorem} 
\usepackage{listings} 
\usepackage{flafter}
\usepackage{indentfirst} 
\usepackage{setspace}
\usepackage{color}
\usepackage{graphicx}
\usepackage{natbib}

\usepackage[hidelinks, colorlinks = False, urlcolor = Sepia, citecolor = MidnightBlue]{hyperref} 
\newtheorem{Theorem}{Theorem}
\newtheorem{Lemma}{Lemma}
\newtheorem{Proposition}{Proposition}

\newtheorem{Definition}{Definition}
\newtheorem{Corollary}{Corollary}
\newtheorem{Example}{Example}
\newtheorem{Assumption}{Assumption}

\title{Convex Order and Moment Persuasion\footnote{We are greatly indebted to Gregory Pavlov, Maria Goltsman, and Charles Zheng for their invaluable guidance and support. We thank Mengxi Zhang and Benny Moldovanu for their generous help during the early stage of this project. We are grateful to Itai Arieli, Ian Ball, Yeon-Koo Che, Kevin He, Elliot Lipnowski, Fedor Sandomirskiy, Roberto Saitto, Mark Whitmeyer, Kai Hao Yang and participants at Midwest International Theory Conference 2026(Toronto), 6th Durham Economics Theory Conference(Durham), Lisbon Meeting 2026(Lisbon), IWGTEA 2026(são paulo) for their helpful comments and suggestions.}}
\author{Xin Wang\footnote{Wang: Department of Economics, Western University, \href{mailto:xwan2552@uwo.ca}{xwan2552@uwo.ca}}}
\date{\today}

\begin{document}

\maketitle
\begin{abstract}
We study Bayesian persuasion with multidimensional states when the sender's payoff depends on the distribution of posterior means. The choice of information policy
in such an environment reduces to an optimization over distributions of posterior means dominated by the prior in convex order. We study extreme points of this feasible set using results on irreducible-component structure from martingale transport theory. When the prior assigns zero mass to the union of the relative boundaries of the non-singleton components, then on each such component the distribution of posterior means is dominated by the prior in convex order. In such cases extreme points must satisfy affine independence of the conditional posterior-mean support in almost every component. Without boundary negligibility, local affine independence remains necessary but is not sufficient. Unlike in the one-dimensional case, in environments with multidimensional states, extreme points can have uncountably many lower-dimensional components that are not singletons.
\end{abstract}

\section{Introduction}
\par
In moment persuasion, a sender chooses an optimal signal to maximize expected payoff when the objective depends on posterior beliefs through their means. The set of feasible signals has a useful convex-order representation: the distributions of posterior means induced by signals are the mean-preserving contractions of the prior. This connection follows from the martingale characterization of convex order in \citet{strassen1965existence} and the comparison-of-experiments perspective of \citet{blackwell1953equivalent}. After \citet{rothschild1971increasing, rothschild1978increasing} introduced this tool to economics, it has recently been revived to link the persuasion framework of \citet{kamenica2011bayesian} to the analysis of posterior means in \citet{kolotilin2017persuasion}, \citet{dworczak2019simple}, and \citet{arieli2023optimal}.
\par
In a moment persuasion model with a compact convex state space and a bounded upper semicontinuous indirect payoff, the sender essentially maximizes an affine upper semicontinuous functional over the set of the mean-preserving contractions of the prior which is a compact convex set. Bauer's maximum principle then implies that at least one optimal posterior-mean distribution is an extreme point. This observation motivates a geometric question: which mean-preserving contractions are extreme? An answer restricts the distributions that need to be considered in optimization, although extremality by itself does not establish optimality for a particular payoff.
\par
In one dimension, integrated-distribution inequalities identify intervals on which information can be pooled and points at which it is fully revealed. This structure underlies the characterizations in \citet{kleiner2021extreme} and \citet{arieli2023optimal}. In several dimensions, pooling can occur on sets of different dimensions, and a partition can contain uncountably many nontrivial components. The irreducible-component structure developed by \citet{de2019irreducible}, together with the related analysis of \citet{obloj2017structure}, provides a way to organize these possibilities. For a given pair of measures in convex order, every martingale coupling sends a posterior mean only to states in the convex closure of its component. We use this existing decomposition as the starting point for the extreme-point analysis.
\par
Theorem \ref{Th1} provides a complete geometric characterization of mean-preserving contractions in multidimensional settings. We show that whenever a measure $\mu$ is dominated by $\mu_0$ in convex order, the state space admits a canonical partition into convex, relatively open cells such that the restriction of $\mu$ to each cell remains dominated by the corresponding restriction of $\mu_0$. Within each cell, any feasible martingale transport must map states into the closure of that same cell, a concentration property that parallels the irreducible-component structure identified in the martingale transport literature. This decomposition generalizes the one-dimensional partition into intervals and singleton points and provides the geometric backbone for our analysis of extreme points.
\par
To obtain a sharp and tractable characterization as Theorem \ref{Th1}, we impose a mild regularity condition—Assumption \ref{As1}—which requires that the relative boundaries of the convex cells be negligible under the prior. Assumption \ref{As1} does not directly impose any assumption on the prior. Though absolute continuity of the prior does not guarantee this assumption in general (because the convex partition may contain uncountably many cells), this assumption still holds in many economically relevant environments, including those where the partition arises from the subdifferential of a convex function, or in two-dimensional settings where Lebesgue measurability ensures boundary negligibility.
\par
Theorem~\ref{Th2} characterizes extremality under the same boundary condition. A feasible distribution is extreme if and only if its conditional support is affinely independent in almost every component. Thus a component of dimension $d$ contains at most $d+1$ conditional posterior means, while a singleton component reveals its state. The proof first establishes the result on a single component by constructing a bounded signed perturbation and controlling it uniformly by the convex-order slack. A measurable disintegration argument then passes from the local statements to the global characterization. 
\par
Section \ref{sec4} examines the limits of this characterization when Assumption \ref{As1} fails. In such cases, the relative boundaries of the convex cells may carry positive prior mass, allowing martingale transports to exploit boundary interactions and mix mass across otherwise distinct concentration regions. Although a version of the local affine-independence condition remains necessary for extremality, it is no longer sufficient: boundary points introduce additional degrees of freedom that can break global extremality even when each local component is well-behaved. This analysis clarifies the precise role of Assumption \ref{As1} in our geometric characterization and highlights the subtleties that arise in multidimensional persuasion when boundary effects cannot be ruled out.
\par
The distinction between local and global support matters for economic applications. A distribution of posterior means may have finite support within each component and still have an uncountable aggregate support. Our product-recommendation example has a uniform prior on the unit square and payoff $U(x)=x_1x_2$. Revealing the sum of the two coordinates induces the unique optimal posterior-mean distribution. This distribution is nonatomic, yet each of its uncountably many one-dimensional pooling components contains a single posterior mean. An explicit payoff identity and a convex-majorant certificate verify optimality and uniqueness.
\par
The paper extends the results of \citet{kleiner2024extreme}, who show that every finitely supported extreme mean-preserving contraction of an absolutely continuous prior admits a convex partition such that, on each cell, the mean-preserving contraction has affinely independent support. In addition to this necessary condition, they provide a sufficient criterion for testing whether a given finitely supported mean-preserving contraction is extreme. Our contribution is to establish a necessary and sufficient characterization of extremality under Assumption \ref{As1}. This assumption is substantially weaker than \citet{kleiner2024extreme}'s assumptions of finite support and absolute continuity of the prior, since we allow more general prior and infinite support. Moreover, Example \ref{example3} demonstrates that our convex partition can be strictly coarser than theirs, implying that extremality may require affine independence on larger sets than those identified by their partition.
\par
This paper also complements the dual approach of \citet{dworczak2024persuasion}. They develop a persuasion duality using optimal prices and their contact sets to characterize optimal signals and the associated pooling regions. Their Theorem 8 establishes that, under differentiability of the sender's objective and continuity of the prior, there exists an optimal signal whose posterior means within each price-induced pooling region are extreme points of their own closed convex hull. In multiple dimensions, however, this condition permits infinitely many posterior means within a single region. They explicitly suggest that a finer partition derived from martingale transport might yield a bound of (d+1) posterior means per region, leaving this question for future research. Our characterization provides a conditional resolution of this proposed refinement: whenever an extreme optimal distribution satisfies our boundary assumption, its canonical components satisfy the stronger requirement of affine independence. More broadly, our criterion applies to feasible distributions independently of the sender's objective or the construction of an optimal dual price.
\par
Our paper contributes to the growing literature on the extreme-point approach in economic theory (e.g., \citet{yang2023extreme}, \citet{lahr2024extreme}, \cite{yang2025multidimensional}). Among these studies, the most closely related is \citet{yang2026stochastic}, which establishes connections between stochastic optimization, extreme-point decompositions, and order-preserving couplings. Their equivalence theorem identifies closure of the test-function cone under pointwise minima as the key condition for a general pointwise decomposition theorem, with applications to multidimensional mean-preserving spreads and stochastic dominance. However, the class of convex functions that characterises mean-preserving contractions is not closed under pointwise minima. Consequently, their general decomposition framework does not apply to the feasible set studied in this paper.
\par
Section~\ref{sec2} develops the conditional decomposition and the boundary assumption. Sections~\ref{sec3} and \ref{sec4} study extremality with and without that assumption. Section~\ref{sec5} derives the implementation result from the polar-set theorem of \citet{de2019irreducible}: in a single component with negligible boundary, polar sets are exactly those covered by two marginal null strips. A countable-mixture argument then yields a coupling whose conditional supports equal the support of the prior almost everywhere. Section~\ref{sec6} applies the results to moment persuasion.

\section{Characterization}\label{sec2}
\par
\subsection{ Preliminaries and Notation}
For a set $A \subset \mathbb{R}^n(n \geq 1)$, let $ri(A)$ denote its relative interior, $rbd(A)$ denote its relative boundary, $cl(A)$ denote its closure, and $conv(A)$ denote its convex hull. For a probability measure $\mu$ on $\mathbb{R}^n$, let $supp(\mu)$ denote its support, and $\hat{supp}(\mu):= cl (conv (supp (\mu)))$ be the smallest convex closed set containing support of $\mu$. $\delta_x$ denotes a Dirac measure concentrated on $x \in \mathbb{R}^n$. For measurable $A \subset \mathbb{R}^n$, $\mu_A$ denotes the restriction of $\mu$ to $A$.
\par 
Unless otherwise stated, $X=\Omega$ is a nonempty compact convex subset of $\mathbb R^n$. We say a collection of subsets of $X$, $\{A_x\}_{x\in X}$, is a \textbf{convex partition} of $X$ if (1)For any $x \in X$, $x \in A_x$ and $A_x$ is a convex set, (2) $\cup_{x\in X} A_x$=$X$, (3)  For any $x,y \in X$, either $A_x=A_y$, or $A_x\cap A_y=\emptyset$. 
\par
We mainly consider the convex partition $\{A_x\}_{x\in X}$ of $X$ with each $A_x$ being relatively open in this paper. The measurability of $\{A_x\}_{x\in X}$ means measurability of the label map $K(x)=A_x$ into the space of nonempty relatively open convex subsets, endowed with its Wijsman topology.\footnote{The set $\hat{\mathcal{K}}$ of all convex closed subsets of $X$ is a Polish space when endowed with the Wijsman topology, and it can be shown that $\hat{\mathcal{K}}$ is isomorphic to $ri \hat{\mathcal{K}}:= \{riK : K \in  \hat{\mathcal{K}}\}$, see \citet{de2019irreducible} for an explanation of the Wijsman topology} Analytically measurable maps are sufficient for this paper, but a Borel version can be chosen. If $\lambda = K\#\mu$, meaning $\lambda(A)=\mu(\{x:K(x) \in A\})$ for a measurable $A \subset \{A_x\}_{x \in X}$, then disintegration gives \[\mu=\int_{\{A_x\}_{x \in X}}\mu_{\mid A_x}\lambda(dA_x)\]
$\mu_{\mid A_x}(\{y:K(y)=A_x\})=1$ $\lambda$-almost everywhere. This defines conditional probabilities even when every individual cell has zero $\mu$-mass. If $\mu(A_x) > 0$, the conditional probability for that cell is $\frac{\mu_{A_x}}{\mu(A_x)}$.
\par
Consider two probability measures $\mu$ and $\mu_0$ on $X$ with finite first-order moment, we say $\mu$ and $\mu_0$ are in the \textbf{convex order}, written as $\mu \preccurlyeq \mu_0$, if $\int \varphi d\mu \leq \int \varphi d\mu_0 $ for all convex functions $\varphi :X \rightarrow \mathbb{R}$. Let $MPC(\mu_0)$ denote the set of all probability measures that are dominated by $\mu_0$ in the convex order. The following standard compactness property follows from the fusion
framework of \citet{elton1992fusions}.
\begin{Lemma}\label{lm5}
If $\Omega\subset\mathbb R^n$ is compact and convex and $\mu_0\in\Delta(\Omega)$, then $MPC(\mu_0)$ is a nonempty convex compact subset of $\Delta(\Omega)$ in the topology of weak convergence.
\end{Lemma}

\par
Define $\Pi(\mu, \mu_0):=\{p \in \Delta(X \times \Omega) \mid p_X=\mu, p_\Omega=\mu_0\},$ which contains all the possible probability measures on $X \times \Omega$ given the marginal on $X$ is $\mu$ and the marginal on $\Omega$ is $\mu_0$. The subset of \textbf{martingale transports} between $\mu$ and $\mu_0$ is $$\mathcal{M}(\mu, \mu_0):=\{p \in \Pi(\mu, \mu_0) \mid \mathbb{E}[Y\mid Z]=Z \text{ for }(Z,Y)\sim p\}.$$ Notice that by \citet{strassen1965existence}, $\mathcal{M}(\mu, \mu_0)$ is nonempty if and only if $\mu \in MPC(\mu_0)$.

For any given $x \in X$  and $p \in \mathcal{M}(\mu, \mu_0)$, the conditional kernel $p_x$ is defined $\mu$-a.e. by
\[p(dx,d\omega)=\mu(dx) p_x(d\omega)\]
\par
A finite collection $\{x_1,..., x_k\} \subset \mathbb{R}^n$ is \textbf{affinely independent} if the unique solution to $\sum^k_{i=1}\beta_ix_i=0$ and $\sum^k_{i=1}\beta_i=0$ is $\beta_1=\beta_2=...=\beta_k=0$. A set of affinely independent vectors contains at most $n + 1$ elements.

\subsection{Geometric Characterization}
\par
To motivate the multidimensional characterization, recall the elegant one-dimensional description of mean-preserving contractions. Assume $\Omega = [0,1]$, a distribution $G \in \Delta([0, 1])$ is a mean-preserving contraction of $F \in \Delta([0, 1])$ if and only if $\int^x_0 G(x) dx \leq \int^x_0 F(x) dx$ for all $x \in [0, 1]$ and equality holds at $x=1$.

If $F$ is a continuous distribution, accoridng to the proof of Proposition 1 in \citet{arieli2023optimal}, these inequalities and equalities partition the state space into “local” intervals and singleton points:
\begin{itemize}
\item On each ``local'' interval $(a,b) \subset [0, 1]$, $\int^b_a G(x) dx = \int^b_a F(x) dx$, while $\int^c_a G(x) dx < \int^c_a F(x) dx$ for all $c \in (a,b)$;
\item At each singleton $x$,$\int^x_0 G(x) dx = \int^x_0 F(x) dx$.
\end{itemize}
Based on this ``local'' structure, \citet{arieli2023optimal} prove that the extremality of $G$ implies the support of $G$ contains at most 2 points on each ``local'' interval. 
\par
In higher dimensions, no direct analogue of these integral constraints exists. The challenge is therefore to identify an appropriate multidimensional analogue of these “local” intervals. This motivates the paper of \citet{de2019irreducible} on the martingale transport, which provides the structural foundation for our analysis.
\begin{Proposition} \label{pr1}
      (De March and Touzi 2019) There exists $\hat{p} \in \mathcal{M}(\mu,\mu_0)$ such that for all $p \in \mathcal{M}(\mu, \mu_0)$, $\hat{supp} (p_x)  \subset \hat{supp} (\hat{p}_x)$ for $\mu$-a.s. x.
      Furthermore $\hat{supp} (\hat{p}_x)$ is $\mu$ -a.s. unique, and we may choose this kernel so that:
      \begin{itemize}
      
      \item $K: x \rightarrow \hat{supp} (\hat{p}_x)$ is analytically measurable $X \rightarrow \hat{\mathcal{K}}$, 
     \item  $x \in A_x:= ri\hat{supp} (\hat{p}_x)$ for all $x \in X$, and $\{ri\hat{supp}(\hat{p}_x), x\in X\}$ is a convex partition of $\Omega$.     
      \end{itemize}
\end{Proposition}
\par
Proposition \ref{pr1} states that the state space can be partitioned into convex cells such that every martingale transport must map each state $x$ into the closure of its own cell. These are the multidimensional analogues of the ``local intervals'' in the one-dimensional case, which called irreducible components by \citet{de2019irreducible}. A direct corollary of Proposition \ref{pr1} is that $\mu$ coincides with $\mu_0$ on all $A_x$ that are singletons and not on the relative boundary of other cells of the convex partition.
\begin{Corollary}\label{Co1}
$\mu_{\mid \{ \mid X \backslash cl(A)\} }= \mu_{0 \mid \{X \backslash cl(A)\}}$, where $A:= X \backslash \{x \in X: A_x=\{x\}\}$.
\end{Corollary}
\textbf{Proof:}  By Proposition \ref{pr1}, $\{A_x\}$ form a convex partition of $X$ such that $\hat{p}_x(cl(A_x))=1$ for $\mu$-a.s. x. If $x \in  X \backslash A$, then $A_x=\{x\}$ so that $\hat{p}_x=\delta_x$. If $x \in A$, then $A_x \subset A$ so that $ supp(\hat{p}_x) \subset cl(A)$. 
\par
If $B \subset X$ is a Borel set, then $\mu_0(B)= \int_{X} \mu(dx) \hat{p}_x(B)= \int_A \mu(dx) \hat{p}_x(B \cap cl(A))+ \int_{X \backslash A}\mu(dx) \delta_x(B)$. 
Moreover, if $B \subset X \backslash cl(A)$,  $\mu_0(B)= \int_{X \backslash A}\mu(dx) \delta_x(B)=\mu(B)$ $\hfill\blacksquare$ 
\par
If $n=1$ and $\mu_0$ is a continuous probability measure, then in $\mu_0$-a.s. sense, $X \backslash cl(A)$ is set of singleton points, and all the non-singleton $A_x$ are exactly the ``local'' intervals. The following example illustrates what the multidimensional analogues of the  ``local intervals'' can look like in two dimensions.

\begin{Example}\label{example1}
\begin{figure}[htbp]\caption{Example 1}
\centering
\includegraphics[width=0.5\textwidth]{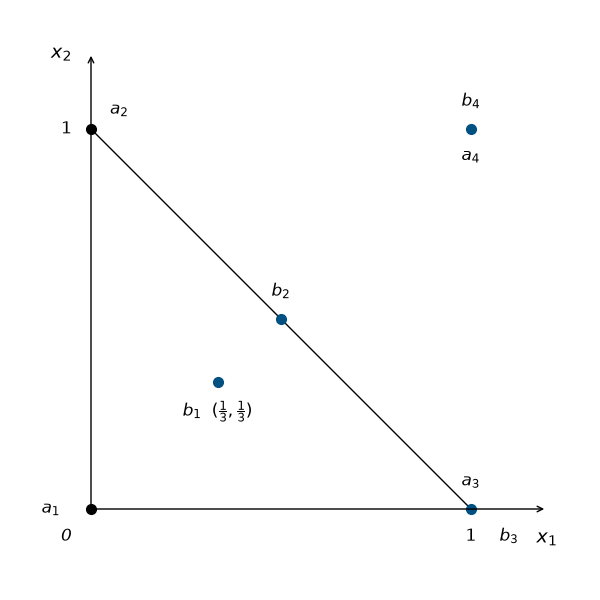}
\end{figure}
Let $X=[0,1]^2$, and let
\[
\begin{gathered}
a_1=(0,0),\ a_2=(0,1),\ a_3=(1,0),\ a_4=(1,1),\\
b_1=(1/3,1/3),\quad b_2=(1/2,1/2),\quad b_3=a_3,\quad b_4=a_4.
\end{gathered}
\]

$\mu_0$ assigns masses $1/12,5/24,11/24,1/4$ at $a_1,a_2,a_3,a_4$, respectively, and $\mu$ assigns mass $1/4$ at each $b_i$. The unique martingale coupling has kernels
\[
p_{b_1}=\tfrac13(\delta_{a_1}+\delta_{a_2}+\delta_{a_3}),\quad
p_{b_2}=\tfrac12(\delta_{a_2}+\delta_{a_3}),\quad
p_{b_3}=\delta_{a_3},\quad p_{b_4}=\delta_{a_4}.
\]
The non-singleton cells are
\[
A_{b_1}=\operatorname{ri}\operatorname{conv}\{a_1,a_2,a_3\},
\qquad A_{b_2}=\operatorname{ri}\operatorname{conv}\{a_2,a_3\}.
\]
The remaining points may be assigned singleton cells. The singleton $a_4$ lies outside the closure of the non-singleton region, so $\mu(\{a_4\})=\mu_0(\{a_4\})$. In contrast, $a_3$ is a boundary point of both displayed cells and $\mu(\{a_3\})\ne\mu_0(\{a_3\})$.

\end{Example}

\par
To simplify our language in describing the local structure of all $\mu \in MPC(\mu_0)$, we introduce the following definition.
\begin{Definition}
Given a convex and relatively open set $A \subset X \subset \mathbb{R}^n$, $\mu_{\mid A}$ is \textbf{concentrated} on $cl(A)$ if
\begin{itemize}
\item for all $p \in \mathcal{M}(\mu, \mu_0)$ and for $\mu_{\mid A}$-a.s. $x$, $\hat{supp} (p_x)  \subset cl(A)$,
\item there exists $\hat{p} \in \mathcal{M}(\mu, \mu_0)$ such that $\hat{supp} (\hat{p}_x)=cl(A)$ for $\mu_{\mid A}$-a.s. $x$.
\end{itemize}
\end{Definition}
In particular, if there is only one n-dimensional concentration region, we say $\mu$ is concentrated on $X$ if $\mu_{\mid ri(X)}$ is concentrated on $X$.
\par
To streamline the geometric characterization, we introduce a mild regularity condition ensuring that the relative boundaries of the convex cells in the partition are negligible. Recall that Proposition \ref{pr1} provides a convex partition $\{A_x\}_{x \in X}$ of the state space, where each $A_x=ri(\hat{supp}(\hat{p}_x))$
for a canonical martingale transport $\hat{p}$. The potential difficulty arises when the relative boundary of one cell intersects with relative boundary of another cell, creating additional degrees of freedom for martingale transports and complicating the analysis of extreme points. As  an illustration consider point $a_3$ in example \ref{example1}, there exists a martingale transport(in fact, it is the unique element in $\mathcal{M}(\mu, \mu_0)$ provided above) mapping it to  $b_1$, $b_2$, $b_3$, each of which belongs to different concentration regions.
\par
To rule out these boundary-interaction pathologies, we define $\mathcal{V}:= \cup_{\{x \in X: A_x \neq \{x\}\}} rbd(A_x)$, where $\{A_x\}_{x \in X}$ is defined in Proposition \ref{pr1}, and we introduce the following assumption:
\begin{Assumption}\label{As1}
$\mathcal{V}$ is $\mu_0$-null set.
\end{Assumption}
Assumption \ref{As1} means the union of all relative boundaries of non-singleton cells has zero $\mu_0$-measure. This assumption eliminates the possibility that a point lies simultaneously on the boundary of multiple convex cells in a way that affects feasible martingale transports. We delay our more detailed discussion of Assumption \ref{As1} to section \ref{subsec23}. 
\par
Under Assumption \ref{As1} we can combine Proposition \ref{pr1} and Corollary \ref{Co1} to obtain a simple geometric characterization of convex order as follows.
\begin{Theorem}\label{Th1}
     Suppose Assumption \ref{As1} holds, then $\mu \in MPC(\mu_0)$ if and only if there exists a collection of disjoint convex and relatively open sets $\{A_x\}_{x\in X}$ such that for $\mu$-a.s. $x$,
\begin{itemize}
  \item $x \in A_x$,
  \item $\mu_{\mid A_x}$ is concentrated on $cl(A_x)$,
  \item $\mu_{\mid A_x}\preccurlyeq \mu_{0\mid A_x}$,  
\item $\mu_{\mid \{\Omega \backslash A\}}= \mu_{0\mid \{\Omega \backslash A\}}$, where $A=\cup_{\{x \in X: A_x \neq \{x\}\}} A_x$.
\end{itemize}
\end{Theorem}
\textbf{Proof:} Necessary condition:  According to Proposition \ref{pr1}, $\mu_{\mid A_x}$ is concentrated on $cl(A_x)$ and $rbd(A_x)$ is $\mu_0$-null set, it is trivial $\mu_{ \mid A_x}\preccurlyeq \mu_{0 \mid A_x}$. Define $A_i$ as each non-singleton $A_x$ that does not lie on the relative boundary of another cell in Proposition \ref{pr1}. If a non-singleton $A_x$ lies on the relative boundary of another cell, then it is $\mu_0-$ and $\mu-$ null set given Assumption \ref{As1} holds. If $A_x=\{x\}$ and lies on the relative boundary of another cell, then it is either zero measure under both $\mu_0$ and $\mu$ or lies on the relative interior of some cell $A_i$. So the only $A_x$ matters here is the singleton that does not lie on the relative interior of other cells. According to Corollary \ref{Co1}, $\mu_{\mid \{\Omega \backslash A\}}= \mu_{0\mid \{\Omega \backslash A\}}$. 
\par
Sufficient condition: For any convex function $\phi: X \rightarrow R$,  

\[\begin{aligned}\int_X \phi d\mu = & \int_{ X \backslash \cup_{I \in I}A_i} \phi d\mu + \int_{\cup_{I \in I}A_i} \phi d\mu \\
      \leq & \int_{ X \backslash \cup_{I \in I}A_i} \phi d\mu_0 + \int_{\cup_{I \in I}A_i} \phi d\mu_0 \\
      =&\int_X \phi d\mu_0
\end{aligned}\] $\hfill\blacksquare$ 
\par
Theorem \ref{Th1} shows that convex order admits a local decomposition: the global dominance is equivalent to convex dominance on each convex cell of the canonical partition, together with equality outside the non-singleton cells. This mirrors the one-dimensional structure, where the state space decomposes into intervals and singleton points. This local structure is the key to the characterization of extreme points in Section 3.
\subsection{The Generality of Assumption \ref{As1}} \label{subsec23}
\par
Before turning to the characterization of extreme points in Section \ref{sec3}, we discuss the generality of Assumption \ref{As1}, which plays a central role in our geometric decomposition. Since each $A_x$ in the convex partition is a convex set, and since the relative boundary of any convex set has Lebesgue measure zero (\citet{lang1986note}) one might expect Assumption \ref{As1} to hold whenever $\mu_0$ is absolutely continuous with respect to Lebesgue measure. However, the convex partition may contain uncountably many convex sets, and Lebesgue measure is not countably additive over uncountable unions. Consequently, the set $\mathcal{V}$ need not be Lebesgue measurable, and even when it is, it may fail to have measure zero. Our Lemma \ref{lm3} gives an affirmative answer when $n=1,2$ and a negative result for $n \geq 3$.
\begin{Lemma}\label{lm3}
      For $n=1$ and any given convex partition $\{A_x\}_{x \in X}$, $\mathcal{V}:= \cup_{\{x \in X: A_x \neq \{x\}\}} rbd(A_x)$ is zero Lebesgue measure. For $n=2$ and any given convex partition $\{A_x\}_{x \in X}$, if $\mathcal{V}$ is Lebesgue measurable, then $\mathcal{V}$ is zero Lebesgue measure. For $n \geq 3$, there exists a convex partition $\{A_x\}_{x \in X}$ such that $\mathcal{V}$ is compact and has positive Lebesgue measure.
\end{Lemma}

\par
In one dimensional case, convex partition consists of at most countable intervals with each of their boundaries being the two endpoints. Thus $\mathcal{V}$ is zero Lebesgue measure. In two dimensional case, the only problematic cells are one-dimensional ones (line segments). This is because there can be at most countably many two-dimensional cells, and their boundaries contribute zero Lebesgue measure.\footnote{Since $\mathbb{R}^n$ is separable, any family of pairwise disjoint nonempty open subsets is at most countable (see Thm. 2.8 in \citet{oxtoby2013measure}). Because each full-dimensional convex cell contains a nonempty open set (its interior), the convex partition has at most countably many n-dimensional cells.} For the one-dimensional cells, \citet{larman1971compact}'s first lemma on disjoint line segments ensures that the union of their endpoints has measure zero. The proof of Lemma \ref{lm3} carefully handles cases where closures of linear segments intersect, showing that one can enlarge the set of endpoints to a disjoint family of line segments whose endpoints still form a measure-zero set. In contrast, for 
$n \geq 3$, \citet{larman1971compact} shows that one can construct a compact family of disjoint line segments whose set of endpoints has positive Lebesgue measure. Treating all other points as singleton cells yields a convex partition for which $\mathcal{V}$ is compact and has positive measure. This establishes that Assumption \ref{As1} may fail in higher dimensions.\footnote{It is still an open question whether this construction can be realized as a convex partition that satisfies the concentration property for some $\mu$ and $\mu_0$ that are in the convex order.} 
\par
It is well-known that the projection of a piecewise linear convex function onto its domain creates a convex partition of state space. By Rademacher's Theorem (Theorem 10.8 in \citet{villani2008optimal}), we know the subset of the nondifferentiable domain of a convex function is Lebesgue-negligible. This motivates Lemma \ref{lm4}, which provides another way to verify the Assumption \ref{As1}. Given a convex function $f(x)$ defined on $X$, denote $ND_f:=\{x\mid f \text{ is nondifferentiable at } x\}$,
\begin{Lemma}\label{lm4}
If there exists a convex function $f$ such that $\mathcal{V}$=$ND_f \cup\partial X $, then $\mathcal{V}$ is zero Lebesgue measure.
\end{Lemma}
\citet{kleiner2024extreme} consider a special convex partition named power diagram which is essentially a projection of a piecewise linear convex function with each cell of the projection having same dimensions as $X$. This is a special case of Lemma \ref{lm4} since the boundary of each cell in the power diagram is the nondifferentiable region of the piecewise linear convex function plus the boundary of $X$.

\section{Extreme Points}\label{sec3}
\par
We say a probability measure $\mu$ is finitely supported if $supp(\mu)$ is finite. Let $FMPC(\mu_0)$ denote the set of all finitely supported probability measures that belong to $MPC(\mu_0)$. We begin by characterizing the extreme points of $FMPC(\mu_0)$ in the case where all mass lies within a single concentration region.

\begin{Proposition}\label{pr2}
      Suppose $\mu \in FMPC(\mu_0)$ is concentrated on $X$, and Assumption \ref{As1} holds, then $\mu$ is an extreme point of $FMPC(\mu_0)$ only if $\mu$ has affinely independent support.
\end{Proposition}
\textbf{Proof:} Suppose the support of $\mu$ is not affinely independent. Let $\{x_i\mid i \in I\}$ denote the support of $\mu$. According to Corollary \ref{Co2} in Section \ref{sec5}, $\exists p \in \mathcal{M}(\mu, \mu_0)$, $\forall x_i$, $supp(p_{x_i})=supp(\mu_0)$. Let $\mu_i=\mu(x_i)p_{x_i}$, then $\mu_i(X) = \mu({x_i})$, $r(\mu_i)=:\frac{1}{\mu_i(X)}\int_X x d\mu_i(x)=\int_X x dp_{x_i}(x)=x_i$, and $\mu_0 =\sum_i \mu_i$. The remainder of the proof follows the argument of Proposition 1 in \citet{kleiner2024extreme}.\footnote{Notice the procedure of constructing graphs in the proof of Proposition 1 in \citet{kleiner2024extreme} always stops with a graph $G_N$ where $\{x_i \mid i \in C\}$ is affinely dependent, because $supp(p_{x_i})=supp(\mu_0)$ for all $i\in I$, $\{x_i\mid i \in I\}$ are all included in $C$ in the end} $\hfill\blacksquare$ 
\par
The Proposition 1 in \citet{kleiner2024extreme} consider a convex partition which is the finest one preserving convex dominance on each cell. Theorem \ref{Th1} implies that convex dominance also holds on each cell of our (coarser) convex partition.\footnote{Depending on the given $\mu_0$ and $\mu$, our convex partition can be strictly coarser than the convex partition in \citet{kleiner2024extreme}. Example \ref{example3} in Appendix provides an illustration.} Proposition \ref{pr2} therefore shows that within each cell, affine independence is still necessary for extremality of finitely supported measures. 
\par
We can extend Proposition \ref{pr2} to any extreme points of $MPC(\mu_0)$ given there exists only one concentration region. The proof is in the Appendix. 
\begin{Proposition}\label{pr3}
      Suppose $\mu \in MPC(\mu_0)$ is concentrated on $X$, and Assumption \ref{As1} holds, then $\mu$ is an extreme point of $MPC(\mu_0)$ only if $\mu$ has affinely independent support.
\end{Proposition}
\par
Proposition \ref{pr3} is proved by a two-sided perturbation argument. If
$\operatorname{supp}\mu$ is affinely dependent, we construct a nonzero
signed measure $\sigma=h\mu$, with bounded $h$, that preserves total
mass and first moments. For infinite support, this construction uses
positive-mass neighborhoods rather than individual atoms. The crucial
step uses the concentration and boundary assumptions to establish a
uniform estimate
\[
\int_X \varphi\,d(\mu_0-\mu)
\geq c\left|\int_X \varphi\,d\sigma\right|,
\]
where $c>0$ is independent of the continuous convex function
$\varphi$. Thus, one sufficiently small $\varepsilon>0$ makes both
$\mu\pm\varepsilon\sigma$ feasible, contradicting extremality. Although this shares the mass-redistribution intuition of the proof of Proposition 1 in \citet{kleiner2024extreme}, their proof relies on finitely many positive-mass atoms, terminating overlap-graph constructions, and transfers along finite spanning trees. These steps do not directly apply to nonatomic distributions, whose individual support points carry zero mass. The additional argument in Proposition \ref{pr3} is therefore to construct a perturbation of the measure itself and control all convex-order inequalities simultaneously.
\par
Finally, combining proposition \ref{pr3} with Theorem \ref{Th1}, we can give our full characterization of extreme points of $MPC(\mu_0)$. The proof of Theorem \ref{Th2} further shows the nuance of localization structure of any mean-preserving contractions. For the sufficiency proof, we essentially prove that for any two mean-preserving contractions, their convex combination must have a coarser convex partition than theirs. For the necessity proof, If a given probability measure $\mu \in MPC(\mu_0)$ is an extreme point of $MPC(\mu_0)$, then its conditional probability measure to any cell $A_x$ in the convex partition is an extreme point of $MPC(\mu_{0 \mid A_x})$. It allows us to focus on each simpler local mean-preserving contraction.
\begin{Theorem}\label{Th2}
      Suppose $\mu\in MPC(\mu_0)$ and Assumption \ref{As1} holds, then $\mu$ is an extreme point of $MPC(\mu_0)$ if and only if there exists a collection of disjoint convex and relatively open sets $\{A_x\}_{x\in X}$ such that for $\mu$-a.s. $x$,
      \begin{itemize}
        \item $x \in A_x$,
        \item $\mu_{\mid A_x}$ is concentrated on $cl(A_x)$,
        \item $\mu_{\mid A_x}$ has affinely independent support,
        \item $\mu_{\mid A_x}\preccurlyeq \mu_{0\mid A_x}$,
        \item $\mu_{\mid \{\Omega \backslash A\}}= \mu_{0\mid \{\Omega \backslash A\}}$, where $A=\cup_{\{x \in X: A_x \neq \{x\}\}} A_x$.
      \end{itemize}
 \end{Theorem}
 \textbf{Proof:} Sufficient condition:
 If $\mu$ is not a extreme point, then we can find two probability measures $\mu_1,\mu_2 \in MPC(\mu_0)$, such that $\mu=\frac{1}{2}\mu_1+\frac{1}{2}\mu_2$. Define corresponding convex partitions $\{A^1_x\}_{x \in X}$ and $\{A^2_x\}_{x \in X}$ for $\mu_1$ and $\mu_2$ as in Proposition \ref{pr1}. We say $\mu_1$ has finer concentration regions than $\mu$ with respect to $\mu_0$ if $A^1_x \subset A_x$ for $\mu_1$-a.s $x$. According to Lemma \ref{lm9} in Appendix, both $\mu_1$ and $\mu_2$ have finer concentration regions than $\mu$ with respect to $\mu_0$. The supports of $\mu_1$ and $\mu_2$ must be subsets of support of $\mu$ in each $A_x$. $\mu$ has affinely independent supports in each $A_x$. Because convex dominance implies that means must coincide for all these measures on given $A_x$, then $\mu_{\mid A_x}=\mu_{1\mid A_x}=\mu_{2\mid A_x}$, which implies $\mu=\mu_1=\mu_2$. 
 \par
Necessary condition: We only need to prove the third condition, all other conditions come from Theorem \ref{Th1}. According to Proposition \ref{pr1}, the partition we have comes from a measurable mapping, so we have the following disintegration form, \[\mu = \int_{\{A_x\}_{x \in X}}\mu_{\mid A_x}\lambda(d A_x),\]
where each $\mu_{\mid A_x}$ is a conditional probability measure defined on the corresponding convex component,  $\lambda(A')=\mu(\{x:A_x \in A'\})$ for all measurable $A' \subset \{A_x\}_{x\in X}$.
The extremality of $\mu$ requires the extremality of $\mu_{\mid A_x}$ for 
$\lambda$-almost all $A_x$. If $A_x$ is singleton, the third condition is trivial. Then applying Proposition \ref{pr3} to all $A_x$ that are not singletons, we finish the proof. $\hfill\blacksquare$ 

\section{When Assumption \ref{As1} does not hold}\label{sec4}
Assumption \ref{As1} ensures that the relative boundaries of the non-singleton cells in the convex partition are negligible under the prior. This regularity is what allows the global convex-order relation to decompose cleanly into independent local dominance relations on each cell. When Assumption \ref{As1} fails, this clean separation may break down: boundary points may simultaneously belong to the closures of multiple cells, and martingale transports may exploit these boundary interactions to mix mass across otherwise distinct concentration regions. As a result, the geometric characterization in Theorem \ref{Th1} no longer holds in full generality, and the structure of extreme points becomes more subtle.
\par
Even without Assumption \ref{As1}, Proposition \ref{pr1} still provides a canonical convex partition $\{A_x\}_{x \in X}$ of the state space, and Corollary \ref{Co1} continues to identify the fully revealing singleton cells. What fails is the guarantee that the closures of distinct non-singleton cells are disjoint up to $\mu_0$-null sets. Consequently, a mean-preserving contraction may assign mass to boundary points that lie in the closures of multiple cells, creating additional feasible martingale transports that do not respect the local concentration structure.
\par
Lemma \ref{lm8} and Theorem \ref{Th4} establish what can still be said about extreme points of mean-preserving contraction set when Assumption \ref{As1} is not satisfied.
\begin{Lemma}\label{lm8}
Suppose $\mu \in MPC(\mu_0)$ and $\mu_{\mid ri(X)}$ is concentrated on $X$, then $\mu$ is an extreme point of $MPC(\mu_0)$ only if
 $\mu$ has affinely independent support on $ri(X)$.
\end{Lemma}
\textbf{Proof:} If \(\mu(ri(X))=0\), the conclusion is immediate. Hence assume \(\mu(ri(X))>0\). By concentration, there exists a martingale coupling $\pi(dx,dy)=\mu(dx)p_x(dy)$ between \(\mu\) and \(\mu_0\) such that $\hat{supp}(p_x)=X \text{ for }\mu\text{-almost every }x\in ri(X)$. Define probability measures
\[
\beta_0:=\int_{ri(X)} p_x \mu_{\mid ri(X)}(dx).
\]
Then \(\mu_{\mid ri(X)}(dx)p_x(dy)\) is a martingale coupling between $\mu_{\mid ri(X)}$ and \(\beta_0\), so \(\mu_{\mid ri(X)} \preceq \beta_0\).

We first show that extremality of \(\mu\) in \(\operatorname{MPC}(\mu_0)\) implies extremality of \(\mu_{\mid ri(X)}\) in \(MPC(\beta_0)\).
For any \(\beta \preceq \beta_0\), define
\[
\mu_\beta:=\mu(ri(X)) \beta+\mu_{rbd(X)}.
\]
This is a probability measure. For every convex function
\(\varphi:\mathbb R^d\to\mathbb R\), Jensen's inequality gives
\[
\begin{aligned}
\int_X\varphi\,d\mu_\beta
&\leq
\mu(ri(X)) \int_X\varphi\,d\beta
+\int_{X\setminus ri(X)}\varphi(x)\,\mu(dx)\\
&\leq
\int_{ri(X)}\int_X\varphi(y)\,p_x(dy)\,\mu(dx)
+\int_{X\setminus ri(X)}\int_X\varphi(y)\,p_x(dy)\,\mu(dx)\\
&=\int_X\varphi\,d\mu_0.
\end{aligned}
\]
Thus \(\mu_\beta\in\operatorname{MPC}(\mu_0)\). If $\mu_{\mid ri(X)}=\tfrac12\beta_1+\tfrac12\beta_2$, $\beta_1,\beta_2\in\operatorname{MPC}(\beta_0)$, then $\mu=\tfrac12\mu_{\beta_1}+\tfrac12\mu_{\beta_2}$. Extremality of \(\mu\) gives \(\mu_{\beta_1}=\mu_{\beta_2}\), and \(\mu(ri(X))>0\) then implies \(\beta_1=\beta_2\). Hence \(\mu_{\mid ri(X)}\) is extreme in \(\operatorname{MPC}(\beta_0)\).

To apply Proposition \ref{pr3}, we construct an intermediate prior
\(\eta\) with \[\mu_{\mid ri(X)} \preceq \eta \preceq \beta_0,
\qquad \eta(rbd(X))=0,\] whose canonical irreducible component is \(ri(X)\), \(\nu\)-almost everywhere.

Let $Z$ have distribution $\mu_{ri(X)}$, and \(Y\) have distribution \(\mu_{\mid ri(X)}(dx)P_x(dy)\). and let \(T\) be independent of \((Z,Y)\) and uniformly distributed on \((0,1)\). Define $W:=(1-T)Z+TY$ and $\eta$ as the distribution of $W$. Since \(\mathbb E[Y\mid Z]=Z\), independence gives $\mathbb E[W\mid Z]=\tfrac12 Z+\tfrac12\mathbb E[Y\mid Z]=Z$. Therefore \(\mu_{\mid ri(X)} \preceq \eta\). For every convex function \(\varphi:\mathbb R^d\to\mathbb R\),
\[
\begin{aligned}
\int_X\varphi\,d\eta
&=\mathbb E[\varphi((1-T)Z+TY)]\\
&\leq
\tfrac12\mathbb E[\varphi(Z)]
+\tfrac12\mathbb E[\varphi(Y)]\\
&\leq\mathbb E[\varphi(Y)]
=\int_X\varphi\,d\beta_0,
\end{aligned}
\]
where the last inequality uses
\(\mu_{\mid ri(X)} \preceq \beta_0\).
Hence \(\eta\preceq \beta_0\).
\par
Because \(Z\in ri(X)\) almost surely and $(1-t)x+ty\in ri(X)$
for $x \in ri(X)$, $y \in X$, $0<t<1$, we have \(\eta(ri(X))=1\), and consequently \(\eta(\operatorname{rbd}X)=0\).
\par
Let \(Q_x\) be the conditional distribution of \(W\) given \(Z=x\). Then \(Q(dx,dw):=\mu_{ri(X)}(dx)Q_x(dw)\) is a martingale coupling between \(\mu_{ri(X)}\) and \(\eta\). We claim that $\operatorname{supp}p_x\subseteq\operatorname{supp}Q_x \text{ for }\mu_{ri(X)} \text{-almost every }x$. Fix such an \(x\), a point \(y\in\operatorname{supp}P_x\), and \(\varepsilon>0\). Choose \(\delta\in(0,1)\) such that \(\delta\operatorname{diam}(X)<\varepsilon/2\). Whenever \(y'\in X\cap B(y,\varepsilon/2)\) and \(t\in(1-\delta,1)\),
\[\|(1-t)x+ty'-y\|
\leq (1-t)\|x-y'\|+\|y'-y\|
<\varepsilon.
\]
It follows that
\[
Q_x(X\cap B(y,\varepsilon))
\geq
\delta\,P_x(X\cap B(y,\varepsilon/2))
>0.
\]
Thus \(y\in\operatorname{supp}Q_x\), proving the claim.
Consequently, $\hat{supp}(Q_x)=X \text{ for }\nu\text{-almost every }x$. By maximality of the canonical conditional convex supports, the canonical irreducible paving associated with \((\nu,\eta)\) therefore satisfies $A_x=ri(X) \text{ for }\nu\text{-almost every }x$.

\par
All the assumptions of Proposition \ref{pr3} now hold for
\((\nu,\eta)\). Moreover, $\nu\in\operatorname{MPC}(\eta) \subseteq \operatorname{MPC}(\beta_0)$. Since \(\nu\) is extreme in \(\operatorname{MPC}(\beta_0)\), it is also extreme in \(\operatorname{MPC}(\eta)\). Proposition \ref{pr3} implies that \(\operatorname{supp}(\nu)\) is affinely independent. Finally, since \(ri(X)\) is relatively open in \(\operatorname{aff}X\), $\operatorname{supp}(\mu)\cap ri(X) =\operatorname{supp}(\nu)\cap ri(X)$. Therefore \(\mu\) has affinely independent support on \(ri(X)\).
$\hfill \blacksquare$ 

\par
If $\mu$ has more than one concentration region, then we can apply Lemma \ref{lm8} to each concentration region.

\begin{Theorem}\label{Th4}
$\mu\in MPC(\mu_0)$ is an extreme point of $MPC(\mu_0)$ only if there exists a collection of disjoint convex and relatively open sets $\{A_x\}_{x\in X}$ such that for $\mu$ a.s. $x$,
\begin{itemize}
\item $x\in A_x$,
\item $\mu_{\mid A_x}$ is concentrated on $cl(A_x)$,
\item $\mu_{\mid A_x}$ has affinely independent support,
\item $\mu_{ \{\Omega \backslash cl(A)\} }= \mu_{0  \{\Omega \backslash cl(A)\}}$, where $A:= \mathbb{R}^n \backslash \{x \in \mathbb{R}^n: A_x=\{x\}\}$.
\end{itemize}
\end{Theorem}
\textbf{Proof:} We only need to prove $\mu_{\mid A_x}$ has affinely independent support for all $x \in X$, the other conditions come directly from Proposition \ref{pr1} and Corollary \ref{Co1}.
According to Proposition \ref{pr1}, the partition we have come from a measurable mapping, so we have the following disintegration form, $$\mu = \int_{\{A_x\}_{x \in X}}\mu_{\mid A_x}\lambda(dA_x),$$
where each $\mu_{\mid A_x}$ is a conditional probability measure defined on the corresponding convex component,  $\lambda(A')=\mu(\{x:A_x \in A'\})$ for all measurable $A' \subset \{A_x\}_{x\in X}$.
The extremality of $\mu$ requires the extremality of $\mu_{\mid A_x}$ for 
$\lambda$-almost all $A_x$. Then applying Lemma \ref{lm8} to all $A_x$, we finish the proof. $\hfill\blacksquare$ 

\par
Theorem \ref{Th4} therefore provides necessary conditions for extremality even when Assumption \ref{As1} fails. However, these conditions are no longer sufficient. The failure of boundary negligibility allows martingale transports to mix mass across the closures of distinct cells, and this additional flexibility can destroy extremality even when each local component satisfies affine independence. The following example illustrates this phenomenon.
\begin{Example}\label{example2}
Let $X=[0,1]^2$, $U$ is a uniform probability measure on $X$, and let
\[
\begin{gathered}
b_1=(1/2,1/2),\quad b_2=(1/2,0),\quad b_3=(0,1/2),\\
b_4=(1/2,1),\quad b_5=(1,1/2).
\end{gathered}
\]
Define
\[
\mu_0=\tfrac15U+\tfrac15\sum_{i=2}^5\delta_{b_i},
\qquad
\mu=\tfrac35\delta_{b_1}+\tfrac1{10}\sum_{i=2}^5\delta_{b_i}.
\]
\begin{figure}[htbp]\caption{Example 2}
\centering
\includegraphics[width=0.5\textwidth]{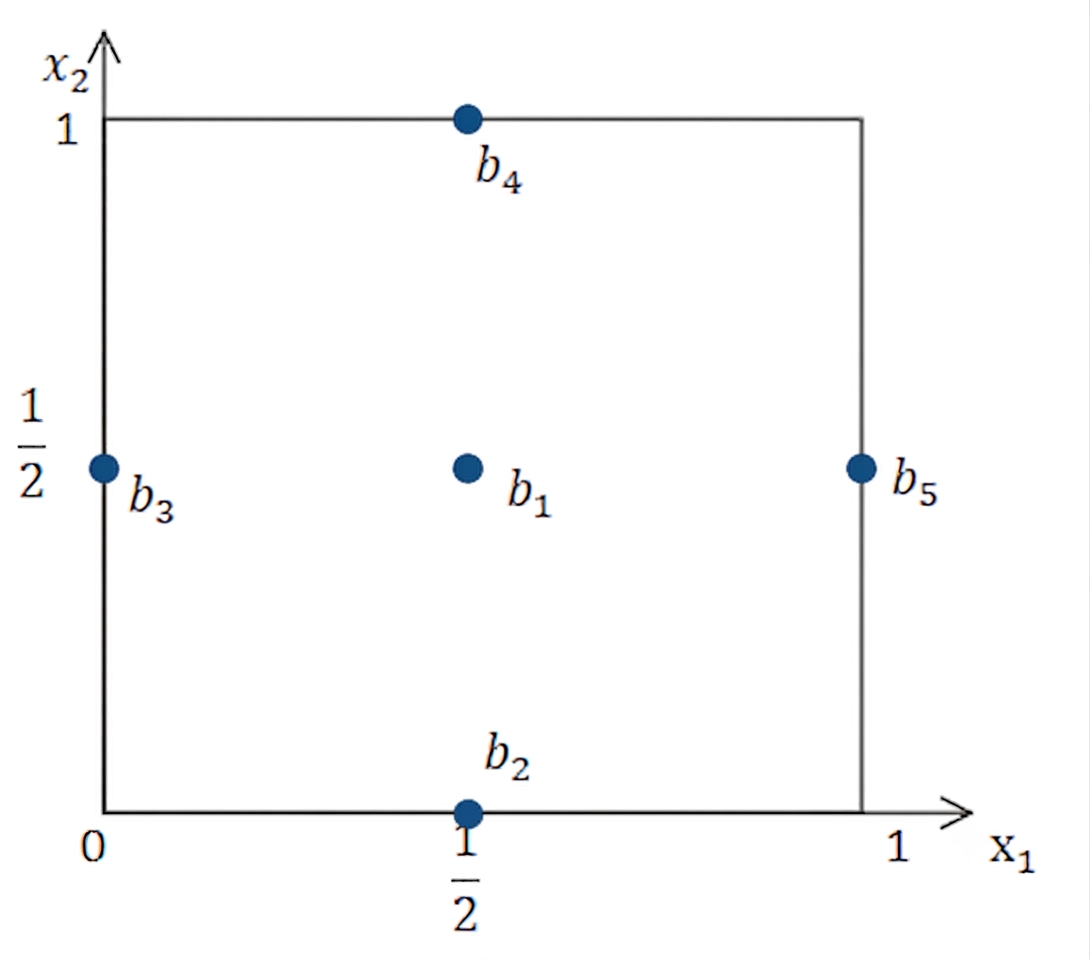}
\end{figure}
\par
$\mu$ and $\mu_0$ are in the convex order. The support of $\mu$ includes 5 points, $b_1$,$ b_2$, $b_3$, $b_4$, $b_5$. Each of them belongs to a different cell of the convex partition: 
\[A_{b_1}=(0,1) \times (0,1), A_{b_2}=\{b_2\},A_{b_3}=\{b_3\},A_{b_4}=\{b_4\}, A_{b_5}=\{b_5\}.\] Assumption \ref{As1} fails in the example,  because relative boundary of $A_{b_1}$ includes points which have positive measure. Although support of $\mu_{\mid A_{b_i}}$ is affinely independent for all i=1,2,3,4,5, $\mu=\frac{1}{2}\delta_{(\frac{1}{2},\frac{1}{2})}+\frac{1}{2}\mu_1$, where $\delta_{(\frac{1}{2},\frac{1}{2})}$ is a Dirac measure concentrated on $(\frac{1}{2},\frac{1}{2})$ and $\mu_1$ assigns 0.2 to all $b_i,i=1,2,3,4,5$. It is easy to verify both $\delta_{(\frac{1}{2},\frac{1}{2})}$ and $\mu_1$ are mean-preserving contractions of $\mu_0$. Thus $\mu$ is not an extreme point.
\end{Example}

\par
This example shows that without Assumption \ref{As1}, affine independence within each cell is no longer sufficient for extremality. Boundary interactions introduce additional degrees of freedom that cannot be ruled out by local geometric conditions alone.

\section{Implementations}\label{sec5}
\par
In this section, we examine the implementation of mean-preserving contractions via martingale transports. We say $p \in \Delta(X \times \Omega)$ implements some $\mu \in MPC(\mu_0)$ if $p \in \mathcal{M}(\mu, \mu_0)$. A Borel set $N\subseteq X\times X$ is
\emph{$\mathcal M(\mu,\mu_0)$-polar} if $p(N)=0$ for every
$p\in\mathcal M(\mu,\mu_0)$. Polar sets describe pairs of posterior means
and states that no feasible implementation can map between. The following theorem reveals the remarkable flexibility of martingale transports:

\begin{Theorem}\label{Th3}
Suppose that $\mu\in MPC(\mu_0)$ is concentrated on $X$ and
$\mu_0(\operatorname{rbd}X)=0$. A Borel set $N\subseteq X\times X$ is
$\mathcal M(\mu,\mu_0)$-polar if and only if there are Borel sets
$N_\mu,N_0\subseteq X$ with $\mu(N_\mu)=\mu_0(N_0)=0$ such that
\[
N\subseteq (N_\mu\times X)\cup(X\times N_0)\footnote{Theorem \ref{Th3} comes directly from Theorem 2.5 in \citet{de2019irreducible} where their $\{Y \notin J(X)\}$ becomes a $\mu_0$-null set under Assumption \ref{As1}.}.
\]
\end{Theorem}
\par
In other words, the only polar sets are those that lie nontrivially outside the cartesian product of supports of $\mu$ and $\mu_0$. Any set intersecting both supports of $\mu$ and $\mu_0$ in a nontrivial way can be transported by some feasible martingale transport. One particularly important consequence of Theorem \ref{Th3} is that it ensures that there exists a martingale transport whose conditional distribution spans the entire support of $\mu_0$ for $\mu\text{-almost every }x$. The existence of this particular martingale transport is the backbone for the proof of Proposition \ref{pr3}.
\par
\begin{Corollary}\label{Co2}
Under the hypotheses of Theorem \ref{Th3}, there is
$p\in\mathcal M(\mu,\mu_0)$ such that
\[
\operatorname{supp}(p_x)=\operatorname{supp}(\mu_0)
\qquad\text{for }\mu\text{-almost every }x.
\]
If $\mu$ is finitely supported, the equality holds at every point of
$\operatorname{supp}(\mu)$.
\end{Corollary}

\noindent\textbf{Proof.}
Let $S=\operatorname{supp}(\mu_0)$ and choose a countable base
$(U_k)_{k\geq1}$ of nonempty relatively open subsets of $S$.
Each $U_k$ has positive $\mu_0$-measure. Fix $k$. For a Borel
disintegration of $Q\in\mathcal M(\mu,\mu_0)$, set
\[
T_k(Q)=\{x:Q_x(U_k)>0\},
\qquad
s_k=\sup_{Q\in\mathcal M(\mu,\mu_0)}\mu(T_k(Q)).
\]
Choose $Q^m$ with $\mu(T_k(Q^m))\to s_k$. The countable mixture
$\bar Q=\sum_{m\geq1}2^{-m}Q^m$ is a martingale transport, and
$T_k(\bar Q)=\bigcup_mT_k(Q^m)$ up to a $\mu$-null set. It follows that
$\mu(T_k(\bar Q))=s_k$. If $s_k<1$, the set
$B=X\setminus T_k(\bar Q)$ has positive $\mu$-measure. By Theorem
\ref{Th3}, some $R\in\mathcal M(\mu,\mu_0)$ satisfies
$R(B\times U_k)>0$. The mixture $(\bar Q+R)/2$ then makes
$T_k$ strictly larger in $\mu$-measure, a contradiction. Thus $s_k=1$;
denote the resulting transport by $P^k$.
\par
Now put $p=\sum_{k\geq1}2^{-k}P^k$. Its conditional distributions
charge every $U_k$ for $\mu$-almost every $x$, simultaneously for all
$k$. Also $p_x(S)=1$ almost everywhere, by the second marginal
constraint. Hence $\operatorname{supp}(p_x)=S$ almost everywhere.
Every support point of a finitely supported $\mu$ is an atom, so none
can belong to a $\mu$-null exceptional set. $\hfill\blacksquare$

\section{Moment Persuasion}\label{sec6}
\par
In moment persuasion, the sender's expected payoff after a signal depends
on the posterior belief only through its mean. Let $\Omega\subset\mathbb R^n$ be compact and convex, let $\mu_0\in\Delta(\Omega)$ be the prior, and let $U:\Omega\to\mathbb R$ denote the sender's indirect payoff. This formulation includes environments in which the receiver's action depends only on the posterior mean and the sender's payoff depends only on that action. It also allows state-dependent realized payoffs when their conditional expectation has the form $U(\mathbb E[\omega\mid S])$.\footnote{A finite vector of moments can be treated as the mean of a transformed state. The prior is then the distribution of that transformed state, with the corresponding compactness and integrability assumptions; see Section 4 of \citet{dworczak2024persuasion}.}
\par
Following the standard argument in \citet{arieli2023optimal}, it is without loss of generality to restrict attention to the unbiased signals where the signal realization space is the same as the posterior mean space and state space, $S=X=\Omega$.  For any unbiased signal, its posterior mean $x:=\mathbb E[\omega\mid s]$ satisfies $\mathbb E[\omega\mid x]=x$. Conversely, any martingale coupling between a posterior-mean distribution and \(\mu_0\) induces an unbiased signal that implements that distribution. Thus the feasible posterior-mean distributions are exactly $MPC(\mu_0)$, and the sender solves
\[
\max_{\mu\in MPC(\mu_0)}\int_\Omega U(x)\,\mu(dx).
\]

\par
Suppose that $U$ is bounded and upper semicontinuous. By the Portmanteau
theorem, $\mu\mapsto\int U\,d\mu$ is upper semicontinuous; it is also
affine. Bauer's maximum principle therefore implies that an optimal
posterior-mean distribution exists and that at least one optimum is an
extreme point of $MPC(\mu_0)$.
\par
This existence statement explains the relevance of the preceding
geometric analysis. When the hypotheses of Theorem \ref{Th2} hold for
an extreme maximizer, its conditional distribution of posterior means
has affinely independent support within almost every component.
Without Assumption \ref{As1}, Theorem \ref{Th4} supplies necessary conditions for extremality. Neither extremality nor those necessary
conditions alone establish optimality for a specified payoff $U$.
Moreover, the components depend on the candidate distribution and the
prior; they do not determine which regions the sender should reveal
before the optimization problem is solved.
\par
A singleton component at $x$ forces its conditional state distribution
to be $\delta_x$, so the state is revealed whenever that posterior mean
occurs. Within a component of dimension $d\geq1$, affine independence
restricts the conditional distribution of posterior means to at most
$d+1$ points. Components of dimension $1\leq d<n$ correspond to pooling on
lower-dimensional sets, a possibility that has no one-dimensional
counterpart. These statements concern conditional distributions on the
components, which may individually have zero prior probability.

\paragraph{An optimality certificate.}
A convex majorant gives a direct way to verify a candidate optimum.
Let $q:\mathbb R^n\to\mathbb R$ be a finite convex function with
$q\geq U$ on $\Omega$. For every feasible $\mu$,
\[
\int U\,d\mu\leq\int q\,d\mu\leq\int q\,d\mu_0.
\]
Hence a feasible $\mu^*$ is optimal whenever
\[
q=U\quad\mu^*\text{-almost surely},
\qquad
\int q\,d\mu^*=\int q\,d\mu_0.
\]
This sufficient condition follows directly from convex order. It is the
weak-duality argument underlying the moment-persuasion approach of
\citet{dworczak2024persuasion}; no general claim of dual attainment is
needed here.

\paragraph{Example: product recommendations and infinitely many pooling regions.}
The following example has a unique optimal distribution of posterior
means and uncountably many nontrivial pooling regions. Consider a
uniformly-distributed prior specialization of the product-information environment in
Section 8.1 of \citet{malamud2021persuasion}.
\par
A platform privately observes
$\omega=(v,p)\sim\mu_0:= U([0,1]^2)$, where $v$ is a
product's value to a consumer and $p$ is the platform's commission if
the consumer purchases. The consumer privately observes a purchase cost
$c\sim U([0,1])$, independent of $\omega$. If a purchase
occurs, the consumer receives $v-c$ and the platform receives $p$;
otherwise both receive zero. Before the state is realized, the platform
commits to a possibly randomized signal based on $(v,p)$. It does not
observe $c$, and its randomization is independent of $c$.
\par
After signal $s$, write
\[
x_1(s):=\mathbb E[v\mid s],\qquad
x_2(s):=\mathbb E[p\mid s].
\]
The consumer purchases if $c\leq x_1(s)$, with the tie having probability
zero. Independence implies
\[
\mathbb E\!\left[p\,\mathbf1_{\{c\leq x_1(s)\}}\mid s\right]
=x_1(s)x_2(s).
\]
If we consider unbiased signal, then the sender's indirect payoff is $U(x)=x_1x_2$.\footnote{The same optimization problem arises with a receiver who chooses $a\in[0,1]^2$ to maximize $-\mathbb E[\|a-\omega\|^2\mid s]$ and a sender whose payoff is $a_1a_2$. The receiver chooses $a=\mathbb E[\omega\mid s]$.} The platform solves
\begin{equation}\label{eq:platform-moment-problem}
\max_{\mu\in MPC(\mu_0)}\int_{[0,1]^2}x_1x_2\,\mu(dx).
\end{equation}

\begin{samepage}
Consider the signal $S^*=T:=v+p$. Exchangeability of $v$ and $p$ gives
the posterior mean
\[
m(t):=\bigl(\mathbb E[v\mid T=t],\mathbb E[p\mid T=t]\bigr)
=\left(\frac t2,\frac t2\right).
\]
\end{samepage}
To prove optimality and identify all optimal posterior means, let $S$
be any feasible signal and write
$M_1:=\mathbb E[v\mid S]$, $M_2:=\mathbb E[p\mid S]$.
Since $M_1+M_2=\mathbb E[T\mid S]$, the identity
$4M_1M_2=(M_1+M_2)^2-(M_1-M_2)^2$ and the conditional variance formula
yield
\begin{equation}
\begin{aligned}
\frac14\mathbb{E}[T^2]-\mathbb{E}[M_1M_2]
&=
\frac14\left(
\mathbb{E}[T^2]
-\mathbb{E}[(M_1+M_2)^2]
\right)
+\frac14\mathbb{E}[(M_1-M_2)^2]
\\
&=
\frac14\left(
\mathbb{E}[T^2]
-\mathbb{E}\!\left[
\bigl(\mathbb{E}[T\mid s]\bigr)^2
\right]
\right)
+\frac14\mathbb{E}[(M_1-M_2)^2].
\end{aligned} \nonumber
\end{equation}
Since
\begin{equation}
\operatorname{Var}(T\mid s)
=
\mathbb{E}[T^2\mid s]
-\bigl(\mathbb{E}[T\mid s]\bigr)^2, \nonumber
\end{equation}
the law of iterated expectations gives
\begin{equation}
\mathbb{E}[\operatorname{Var}(T\mid s)]
=
\mathbb{E}[T^2]
-\mathbb{E}\!\left[
\bigl(\mathbb{E}[T\mid s]\bigr)^2
\right]. \nonumber
\end{equation}
Consequently,
\begin{equation}\label{eq:platform-payoff-gap}
\frac14\mathbb{E}[T^2]-\mathbb{E}[M_1M_2]
=
\frac14\mathbb{E}[\operatorname{Var}(T\mid s)]
+\frac14\mathbb{E}[(M_1-M_2)^2].
\end{equation}
Here
\[
\frac14\mathbb E[T^2]
=\frac14\left(\mathbb E[v^2]+\mathbb E[p^2]+2 \mathbb E[v p]\right)
=\frac14\left(\frac13+\frac13+\frac12\right)=\frac7{24}
\]
is the payoff under $S^*$, and $\mathbb E[M_1M_2]$ is the payoff under any signal $S$. Under $S^*=T$, both nonnegative terms in \eqref{eq:platform-payoff-gap} vanish. $S^*$ is therefore optimal, with expected payoff $7/24$. 

Equality in \eqref{eq:platform-payoff-gap} holds if and only if
$\operatorname{Var}(T\mid S)=0$ and $M_1=M_2$ almost surely. Therefore
every optimal signal satisfies $(M_1,M_2)=m(T)$ almost surely, and the
optimal distribution of posterior means, $\mu^*$, is unique. Because $U$ is continuous, it exposes $\mu^*$ in $MPC(\mu_0)$; in particular, $\mu^*$ is extreme. Since $T$ has a continuous triangular distribution, $\operatorname{supp}\mu^*=\{(z,z):0\leq z\leq1\}$, and $\mu^*$ is nonatomic. No signal with finitely or countably many distinct posterior means can attain the optimum. The uniqueness is at the level of posterior means: additional information can leave these means unchanged. For example, revealing both $T$ and $|v-p|$ also induces $\mu^*$ by symmetry.
\par
The convex-majorant certificate is particularly simple here:
\[
q(x):=\frac{(x_1+x_2)^2}{4},\qquad
q(x)-U(x)=\frac{(x_1-x_2)^2}{4}.
\]
Thus $q$ is convex, coincides with $U$ on the diagonal, and satisfies
$\int q\,d\mu_0=7/24$. It is constant, hence affine, along each pool
of $S^*$, so both inequalities in the certificate bind for $\mu^*$.
\par
For $t\in(0,2)$, the signal $T=t$ pools the states in the segment
\[
C_t:=\{(v,p)\in[0,1]^2:v+p=t\}
=\{(u,t-u):\max\{0,t-1\}\leq u\leq\min\{1,t\}\}.
\]
A version of the conditional prior given $T=t$ is uniform on $C_t$,
whereas the conditional posterior-mean distribution is
$\delta_{m(t)}$. The relatively open segments $A_t:=\operatorname{ri}C_t$
are pairwise disjoint. They are also the canonical cells defined in Proposition \ref{pr1}. Let $(Z,Y)$ be any martingale coupling of $\mu^*$ and $\mu_0$, so that
\begin{equation}
Z\sim\mu^*,\qquad
Y\sim\mu_0,\qquad
\mathbb{E}[Y\mid Z]=Z. \nonumber
\end{equation}
Recall that $q(x)=\frac14(x_1+x_2)^2$. The martingale property implies
$\mathbb{E}[Y_1+Y_2\mid Z]=Z_1+Z_2$. Since $Z_1+Z_2$ is measurable with respect to $Z$, we obtain
\begin{equation}
\begin{aligned}
\mathbb{E}[(Y_1+Y_2)(Z_1+Z_2)]
&=\mathbb{E}\!\left[
(Z_1+Z_2)\mathbb{E}[Y_1+Y_2\mid Z]
\right]\\
&=\mathbb{E}[(Z_1+Z_2)^2]. \nonumber
\end{aligned}
\end{equation}
Consequently,
\begin{equation}
\begin{aligned}
\mathbb{E}[(Y_1+Y_2-Z_1-Z_2)^2]
&=\mathbb{E}[(Y_1+Y_2)^2]
  -\mathbb{E}[(Z_1+Z_2)^2]\\
&=4\left(
\int q\,d\mu_0-\int q\,d\mu^* \nonumber
\right). 
\end{aligned}
\end{equation}
Equality of the two integrals therefore gives $\mathbb{E}[(Y_1+Y_2-Z_1-Z_2)^2]=0$. The integrand is nonnegative, so $Y_1+Y_2=Z_1+Z_2 \text{ almost surely}$. Let $p_z$ be the distribution of $Y$ conditional on $Z=z$. Since the coordinates of $m(t)$ sum to $t$, the preceding almost-sure equality implies $p_{m(t)}(C_t)=1 \text{ for } \mu^* \text{-almost every }t$. 
Because $C_t$ is closed, $\operatorname{supp}p_{m(t)}\subseteq C_t \text{ for } \mu^* \text{-almost every }t$. The sum-revealing implementation attains all of $C_t$ as its conditional support, so it attains the maximal conditional convex support. The endpoints of these segments lie on the boundary of the square, which has zero $\mu_0$-measure. Thus this example exhibits uncountably many lower-dimensional pooling components while satisfying the required boundary negligibility. Each component has one posterior mean, although the aggregate distribution is nonatomic.
\par
Economically, the platform groups high-value, low-commission products
with low-value, high-commission products having the same sum. Expected product value determines the probability of purchase, and the commission determines the platform's payoff if a purchase occurs. The platform wants a high probability of purchase when the commission is high. Pooling opposite types helps achieve this: high-value, low-commission products lend credibility to recommendations that also include low-value, high-commission products. The example also demonstrates why a finite bound on the number of posterior means within each component does not imply a finite bound for the information policy as a whole.

\section{Conclusion}
\par
The irreducible-component structure connects the global convex-order constraint to local restrictions on posterior means. When component boundaries carry no prior mass, the prior and the posterior-mean distribution share a common distribution over cells, and extremality is equivalent to conditional affine independence. Without boundary negligibility, the same local restriction is necessary but need not be sufficient because boundary mass can be allocated across components.
\par
For moment persuasion, this characterization describes the geometry of at least one extreme maximizer under the usual compactness and upper semicontinuity conditions, whenever the relevant boundary hypothesis holds. It does not replace the comparison of payoffs across feasible distributions. The product-recommendation example illustrates both points: a direct optimality certificate identifies a unique, nonatomic posterior-mean distribution, while the geometric decomposition explains its uncountably many one-dimensional pools.
\par
It would be interesting to investigate the following two questions. First, it would be useful to develop conditions on the primitives that guarantee that Assumption \ref{As1} is satisfied. Second, characterizing which extreme points are unique solutions by economically meaningful payoffs would bring the geometric analysis closer to constructive solutions of persuasion problems.


\bibliographystyle{apalike}
\bibliography{ref}

@article{kleiner2021extreme,
  title={Extreme points and majorization: Economic applications},
  author={Kleiner, Andreas and Moldovanu, Benny and Strack, Philipp},
  journal={Econometrica},
  volume={89},
  number={4},
  pages={1557--1593},
  year={2021},
  publisher={Wiley Online Library}
}

@article{arieli2023optimal,
  title={Optimal persuasion via bi-pooling},
  author={Arieli, Itai and Babichenko, Yakov and Smorodinsky, Rann and Yamashita, Takuro},
  journal={Theoretical Economics},
  volume={18},
  number={1},
  pages={15--36},
  year={2023},
  publisher={Wiley Online Library}
}

@article{dworczak2024persuasion,
  title={The persuasion duality},
  author={Dworczak, Piotr and Kolotilin, Anton},
  journal={Theoretical Economics},
  volume={19},
  number={4},
  pages={1701--1755},
  year={2024},
  publisher={Wiley Online Library}
}

@article{malamud2021persuasion,
  title={Persuasion by Dimension Reduction},
  author={Malamud, Semyon and Schrimpf, Andreas},
  journal={arXiv preprint arXiv:2110.08884},
  year={2021}
}

@article{yang2023extreme,
  title={Monotone Function Intervals: Theory and Applications},
  author={Yang, Kai Hao and Zentefis, Alexander K},
  journal={American Economic Review},
  volume={114},
  number={8},
  pages={2239--2270},
  year={2024},
  doi={10.1257/aer.20230330}
}

@article{kamenica2011bayesian,
  title={Bayesian persuasion},
  author={Kamenica, Emir and Gentzkow, Matthew},
  journal={American Economic Review},
  volume={101},
  number={6},
  pages={2590--2615},
  year={2011},
  publisher={American Economic Association}
}

@article{dworczak2019simple,
  title={The simple economics of optimal persuasion},
  author={Dworczak, Piotr and Martini, Giorgio},
  journal={Journal of Political Economy},
  volume={127},
  number={5},
  pages={1993--2048},
  year={2019},
  publisher={The University of Chicago Press Chicago, IL}
}

@article{kolotilin2017persuasion,
  title={Persuasion of a privately informed receiver},
  author={Kolotilin, Anton and Mylovanov, Tymofiy and Zapechelnyuk, Andriy and Li, Ming},
  journal={Econometrica},
  volume={85},
  number={6},
  pages={1949--1964},
  year={2017},
  publisher={Wiley Online Library}
}

@article{obloj2017structure,
  title={Structure of martingale transports in finite dimensions},
  author={Ob{\l}{\'o}j, Jan and Siorpaes, Pietro},
  journal={arXiv preprint arXiv:1702.08433},
  year={2017}
}

@article{de2019irreducible,
  title={Irreducible convex paving for decomposition of multidimensional martingale transport plans},
  author={De March, Hadrien and Touzi, Nizar},
  journal={The Annals of Probability},
  volume={47},
  number={3},
  pages={1726--1774},
  year={2019},
  publisher={JSTOR}
}

@article{blackwell1953equivalent,
  title={Equivalent comparisons of experiments},
  author={Blackwell, David},
  journal={The annals of mathematical statistics},
  pages={265--272},
  year={1953},
  publisher={JSTOR}
}

@article{strassen1965existence,
  title={The existence of probability measures with given marginals},
  author={Strassen, Volker},
  journal={The Annals of Mathematical Statistics},
  volume={36},
  number={2},
  pages={423--439},
  year={1965},
  publisher={Institute of Mathematical Statistics}
}

@article{larman1971compact,
  title={A compact set of disjoint line segments in E3 whose end set has positive measure},
  author={Larman, David G},
  journal={Mathematika},
  volume={18},
  number={1},
  pages={112--125},
  year={1971},
  publisher={London Mathematical Society}
}

@article{lang1986note,
  title={A note on the measurability of convex sets},
  author={Lang, Robert},
  journal={Archiv der Mathematik},
  volume={47},
  number={1},
  pages={90--92},
  year={1986},
  publisher={Springer}
}

@article{kleiner2024extreme,
  title={The extreme points of fusions},
  author={Kleiner, Andreas and Moldovanu, Benny and Strack, Philipp and Whitmeyer, Mark},
  journal={arXiv preprint arXiv:2409.10779},
  year={2024}
}

@book{villani2008optimal,
  title={Optimal transport: old and new},
  author={Villani, C{\'e}dric},
  volume={338},
  year={2008},
  publisher={Springer}
}

@article{elton1992fusions,
  title={Fusions of a probability distribution},
  author={Elton, J and Hill, Theodore P},
  journal={The Annals of Probability},
  pages={421--454},
  year={1992},
  publisher={JSTOR}
}

@article{yang2026stochastic,
  title={Stochastic Optimization and Coupling},
  author={Yang, Frank and Yang, Kai Hao},
  journal={arXiv preprint arXiv:2603.11448},
  year={2026}
}

@book{oxtoby2013measure,
  title={Measure and category: A survey of the analogies between topological and measure spaces},
  author={Oxtoby, John C},
  year={2013},
  publisher={Springer Science \& Business Media}
}

@article{yang2025multidimensional,
  title={Multidimensional Monotonicity and Economic Applications},
  author={Yang, Frank and Yang, Kai Hao},
  journal={arXiv preprint arXiv:2502.18876},
  year={2025}
}

@article{lahr2024extreme,
  title={Extreme points in multi-dimensional screening},
  author={Lahr, Patrick and Niemeyer, Axel},
  journal={arXiv preprint arXiv:2412.00649},
  year={2024}
}

@incollection{rothschild1978increasing,
  title={Increasing risk: I. A definition},
  author={Rothschild, Michael and Stiglitz, Joseph E},
  booktitle={Uncertainty in economics},
  pages={99--121},
  year={1978},
  publisher={Elsevier}
}

@article{rothschild1971increasing,
  title={Increasing risk II: Its economic consequences},
  author={Rothschild, Michael and Stiglitz, Joseph E},
  journal={Journal of Economic theory},
  volume={3},
  number={1},
  pages={66--84},
  year={1971},
  publisher={Elsevier}
}

\newpage
\appendix
\section{Omitted Proofs} 

\noindent\textbf{Proof of Lemma \ref{lm3}.} 
We assume $X=\mathbb R^n$ in this proof, if X is a convex compact subset of $\mathbb R^n$, the proof is similar. Since $\mathbb{R}^n$ is separable, any family of pairwise disjoint nonempty open subsets is at most countable (see Thm. 2.8 in \citet{oxtoby2013measure}). Because each full-dimensional convex cell contains a nonempty open set (its interior), the convex partition has at most countably many n-dimensional cells.  
\par
if $n=1$, we have at most countable many $A_x$ that are not singleton. The relative boundary of each interval $A_x$ is zero Lebesgue measure. Since Lebesgue measure is additive for countable many sets, $\mathcal{V}$ is zero Lebesgue measure. 
\par 
Since the proof for $n=2$ relies on several new established results, we treat this case separately as Lemma \ref{lm10}. See the proof of Lemma \ref{lm10}.  
\par
If $n\geq 3$, by Theorem 1 from \citet{larman1971compact}, there exists a disjoint set $L$ of closed line segments in $\mathbb{R}^n$ such that the Lebesgue measure of the set of end points(the relative boundaries of line segments) is positive, where L(and hence the set of end points of all line segments) is compact. Take the relative interiors of the segments as cells and assign each remaining point of \(\Omega\) to a singleton cell. Together, these cells form a convex partition and $\mathcal{V}$ is compact, has positive measure. $\hfill\blacksquare$ \\

\begin{Lemma}\label{lm10}
Let $\{A_i\}_{i \in I}$ be a pairwise disjoint family of non-singleton, relatively open, convex subsets of $\mathbb{R}^2$.
Suppose that \[\mathcal{V}:=\bigcup_{i \in I}\operatorname{rbd} A_i\] is Lebesgue measurable. Then \[\mathcal L^2(\mathcal{V})=0.\]
\end{Lemma}
\noindent
\textbf{Proof of Lemma \ref{lm10}:} \\

\noindent\textbf{Step 1: The two-dimensional members.}

Let $X_2:=\{x \in X:\dim A_x=2\}$. For $x \in X_2$, relative openness is ordinary openness in
$\mathbb{R}^2$. Hence $\{A_x\}_{x \in X_2}$ is a pairwise disjoint family of nonempty open subsets of $\mathbb{R}^2$. Such a family is countable, since every nonempty open subset of $\mathbb R^2$ contains a point of $\mathbb Q^2$, and distinct members contain distinct rational points. Moreover, if $A \subset \mathbb{R}^2$ is convex with nonempty interior,then \[\mathcal L^2(\partial A)=0\] according to \citet{lang1986note} 

Consequently, $P:=\bigcup_{x \in X_2}\operatorname{rbd}A_x$ is a countable union of null sets, so \[\mathcal L^2(P)=0.\]

\medskip

\noindent
\textbf{Step 2: The one-dimensional members.}

Let $\mathcal A_1:=\{A_x:\dim A_x=1\}$, and define \[E:=\bigcup_{A \in \mathcal A_1}\operatorname{rbd}A.\]
Every $A \in\mathcal A_1$ is, in its affine hull, either an open bounded interval, an open ray, or an entire affine line. We have $\mathcal V=P\cup E$. Define the set of T-junction endpoints by $T:=\{x\in E:\ x\in H\text{ for some }H \in \mathcal{A}_1\}$.

We first prove that \[\mathcal L^{2,*}(T)=0.\]

\medskip
\noindent
\textbf{Step 3: T-junction nullity.}
Fix $x\in T$. Since the members of $\mathcal{A}_1$ are pairwise disjoint, there is a unique $H_x\in \mathcal{A}_1$ such that $x \in H_x$. We call $H_x$ the host of $x$.
Since $x\in E$, choose one $A_x \in \mathcal{A}_1$ such that $x\in\operatorname{rbd}A_x$. Choose a unit vector $v_x$ pointing from $x$ into $A_x$. The direction of $A_x$ cannot be parallel to the affine line $\operatorname{aff}H_x$. Otherwise, since $x$ is a relative interior point of $H_x$, a sufficiently short initial portion of $A_x$ would lie in $H_x$, contradicting pairwise disjointness. We call $A_x$ the terminating segment of $x$. Fix \[\eta:=\frac{1}{10}.\] For every nonparallel pair consisting of an unoriented host direction and an oriented terminating direction, there exists an invertible linear transformation $B$ which sends the first direction to horizontal and the second direction to upward vertical. By continuity, there exist neighborhoods of those two directions such that, under the same transformation $B$, \[|\text{host slope}|\le\eta,\]
and every terminating direction has positive vertical component and satisfies \[\left|\frac{dx}{dy}\right|\le\eta.\] The space of nonparallel direction pairs is second countable. Hence we may choose countably many such direction charts $(U_k,V_k,B_k), k\ge1$, covering all possible direction pairs.

Let $T_k$ consist of those $x\in T$ whose host direction and chosen terminating direction belong to the $k$-th chart. Then $T=\bigcup_{k=1}^\infty T_k$. For $x\in T_k$, put \[p_x:=B_k x.\] Since $x$ is a relative interior point of $H_x$, there exists $a_x>0$ such that $B_k H_x$ contains its affine-line portion over every first coordinate within distance $a_x$ of $(p_x)_1$. Similarly, since the transformed terminating direction has positive vertical component, there exists $b_x>0$ such that $B_k A_x$ extends through every positive vertical displacement less than $b_x$.

For $m\ge1$, define \[T_{k,m}:=\left\{x\in T_k: a_x>\frac1m,\quad b_x>\frac1m \right\}. \] Since $a_x,b_x>0$,$T=\bigcup_{k,m=1}^\infty T_{k,m}$. It therefore suffices to prove \[\mathcal L^{2,*}(T_{k,m})=0\] for every $k,m$. Fix $k,m$, put \[\rho:=\frac1m,\qquad r:=\frac{\rho}{10},\] and work in the transformed coordinates. Partition $\mathbb {R}^2$ into half-open squares of side $r$. We claim that all points of $B_k(T_{k,m})$ lying in the same square have the same host. Take \[p=B_k x, q=B_k y,\qquad x,y \in T_{k,m} \] in the same square. Then $|p_1-q_1|\le r, \qquad |p_2-q_2|\le r$. Let \[G_p:=B_k H_x,\qquad G_q:=B_k H_y.\]
Suppose that $H_x\neq H_y$. Their affine hulls cannot coincide. Indeed, both hosts contain their affine-line portions over horizontal intervals of radius $\rho$ around $p$ and $q$, respectively, while $|p_1-q_1|<\rho$. If their affine hulls coincided, the two host portions would overlap, contradicting pairwise disjointness.
Thus the two affine hulls are distinct almost-horizontal lines. Over the horizontal interval between $p_1$ and $q_1$, both actual host sets are present. They cannot intersect there. Hence one lies strictly above the other. Assume first that the host of $q$ lies above the host of $p$. Then $G_q(p_1)>p_2$.
Since $q\in G_q$ and the slope of $G_q$ has absolute value at most $\eta$, \[0<G_q(p_1)-p_2\le |q_2-p_2|+\eta|q_1-p_1| \le (1+\eta)r.\]

Follow the terminating member $B_k A_x$ upward from $p$ through vertical distance $3r$. This is possible because $3r<\rho<b_x$.
Let the resulting point be $z$. Then $z_2=p_2+3r$.
Since $\left|\frac{dx}{dy}\right|\le\eta$,
we have $|z_1-p_1|\le3\eta r$.
Therefore \[|z_1-q_1| \le |z_1-p_1|+|p_1-q_1| \le (1+3\eta)r = 1.3r < \rho.\]
Hence the point of the affine line of $G_q$ with first coordinate $z_1$ actually belongs to $G_q$.

Moreover, \[\begin{aligned}
G_q(z_1) &\le G_q(p_1)+\eta|z_1-p_1|\\
&\le p_2+(1+\eta)r+3\eta^2r.
\end{aligned} \]
Since $\eta=1/10$,\[(1+\eta)+3\eta^2=1.13<3.\]
Thus \[G_q(z_1)<p_2+3r=z_2.\]

At $p$, the terminating member lies below $G_q$, while at $z$ it lies above $G_q$. By continuity, the portion of $B_k A_x$ from $p$ to $z$ intersects $G_q$. The crossing occurs strictly after $p$ and before vertical displacement $3r<\rho$, so it lies in $B_k A_x$. It also lies within horizontal distance less than $\rho$ of $q$, so it lies in the relatively open host $G_q$.

Finally, $A_x$ and $H_y$ must be distinct. If $A_x=H_y$, then $p$ would lie on the affine line of $H_y$, and since $|p_1-q_1|<\rho$ while $H_y$ contains its affine-line portion over horizontal radius $\rho$ around $q$, we would have \[p\in B_k H_y=B_k A_x,\] contradicting the fact that $x$ is a relative-boundary point of the relatively open set $A_x$.
Hence two distinct members of the original family intersect, a contradiction. The opposite vertical ordering is treated symmetrically. Therefore all points of $B_k(T_{k,m})$ in one grid square have the same host.
\par
For every grid square $Q$, there is consequently an affine line $L_Q$
such that \[B_k(T_{k,m})\cap Q\subset L_Q.\]
There are countably many grid squares, so $B_k(T_{k,m})$ is contained in a countable union of affine lines. Therefore $\mathcal L^{2,*}(B_k(T_{k,m}))=0$. Since $B_k^{-1}$ is Lipschitz, $\mathcal L^{2,*}(T_{k,m})=0$.
Taking the countable union gives \[\boxed{\mathcal L^{2,*}(T)=0.}\]

\medskip

\noindent
\textbf{Step 4: A one-sided Klee--Martin lemma.}

We prove the following lemma.
\begin{Lemma}
Let $F\subset\mathbb R^2$ be compact. Suppose that for every $x\in F$ there exists a nondegenerate closed segment $S_x=[x,x^*]$ having $x$ as one endpoint, and suppose that $S_x\cap S_y=\varnothing (x\neq y)$ Then $\mathcal L^2(F)=0$.
\end{Lemma}

Suppose, toward a contradiction, that $\mathcal L^2(F)>0$. For $x\in F$, let \(u_x:=\frac{x^*-x}{|x^*-x|}\). Fix \(\varepsilon:=\frac1{100}\). For distinct $p,q\in\mathbb Q^2$, set \(D:=|p-q|\) and define \[W(p,q):=\left\{x\in F:|x-p|<\varepsilon D,\quad |x^*-q|<\varepsilon D \right\}.\] We claim that $F=\bigcup_{\substack{p,q\in\mathbb Q^2\\p\neq q}}W(p,q)$.

Indeed, fix $x\in F$ and let \(\ell:=|x^*-x|>0\). Choose rational $p,q$ sufficiently close to $x,x^*$, respectively. Then $D=|p-q|$ is arbitrarily close to $\ell$, so both inequalities defining $W(p,q)$ hold. Since this is a countable cover and $\mathcal L^{2,*}(F)>0$, there exist $p,q$ such that $a:=\mathcal L^{2,*}(W(p,q))>0$. Fix such a pair and write $W:=W(p,q)$.

For every $x\in W$, $|x^*-x| \ge D-|x-p|-|x^*-q| > (1-2\varepsilon)D$. Let \[L:=(1-2\varepsilon)D.\] Thus every segment $S_x$, $x\in W$, has length greater than $L$. Moreover, $|(x^*-x)-(q-p)|<2\varepsilon D$. Therefore all the unit directions $u_x$ lie in a sufficiently narrow cone about the direction $q-p$. For $\varepsilon=1/100$, in particular, $\angle(u_x,u_y)<\frac{\pi}{3} (x,y\in W)$.

Choose $0<t<\frac L4$ and define \[\Phi_t(x):=x+t u_x, \qquad x\in W. \]

We claim that
$|\Phi_t(x)-\Phi_t(y)| \ge \frac14|x-y| (x,y\in W)$. Take distinct $x,y\in W$, and set \[u:=u_x,\qquad v:=u_y,\qquad h:=y-x,\] and \[w:=\Phi_t(y)-\Phi_t(x) = h+t(v-u).\]

Suppose, toward a contradiction, that $ |w|<\frac14|h|$. If $u=v$, then $w=h$, which is impossible. Hence $u\neq v$. Let $\theta=\angle(u,v)$. Since $\theta<\pi/3$, $|\det(u,v)|=|u-v|\cos\frac{\theta}{2} \ge \frac12|u-v|$. Also, $t|u-v|=|w-h| \ge |h|-|w| > \frac34|h|$. Hence $|\det(u,v)|>\frac{3}{8t}|h|$.

The two affine lines containing $S_x$ and $S_y$ meet at a unique point,
which we write as $x+s u=y+r v$. Taking determinants gives \[s-t= \frac{\det(w,v)}{\det(u,v)} \]
and
\[r-t=\frac{\det(w,u)}{\det(u,v)}.\]
Therefore $|s-t| \le \frac{|w|}{|\det(u,v)|} < \frac23t$, 
and likewise $ |r-t|<\frac23t$. Consequently $\frac t3<s<\frac{5t}{3},\frac t3<r<\frac{5t}{3}$. Since $t<L/4$, $0<s,r<L$. But both $S_x$ and $S_y$ have length greater than $L$, so $x+s u=y+r v$ lies in both segments, contradicting pairwise disjointness. Thus $|\Phi_t(x)-\Phi_t(y)| \ge \frac14|x-y|$.

Let \(W_t:=\Phi_t(W)\). Then $\Phi_t$ is injective and its inverse $\Phi_t^{-1}:W_t\to W$ is $4$-Lipschitz.

We now use the following elementary fact: if $g$ is $4$-Lipschitz on a subset of $\mathbb R^2$, then \[\mathcal L^{2,*}(g(B)) \le 128\,\mathcal L^{2,*}(B) \] for every subset $B$ of its domain. Indeed, cover $B$ by squares of side lengths $s_i$. The image of the part lying in one such square has diameter at most $4\sqrt2\,s_i$, and can therefore be covered by a square of area at most $128s_i^2$. Taking infima over all square covers gives the estimate.

Applying this to $\Phi_t^{-1}$ yields
$a=\mathcal L^{2,*}(W) \le 128\,\mathcal L^{2,*}(W_t)$.
Hence $\mathcal L^{2,*}(W_t)\ge\frac{a}{128}$.
Now define \[N_t(F):= \{z\in\mathbb R^2:\operatorname{dist}(z,F)\le t\}.\] Since $|\Phi_t(x)-x|=t$, we have $W_t\subset N_t(F)$. Moreover, $W_t\cap F=\varnothing$. Indeed, if $z=\Phi_t(x)\in F$, then $z\neq x$ and $z\in S_x$. Since $z\in F$, there is also a segment $S_z$ rooted at $z$, and hence $z\in S_x\cap S_z$, contradicting pairwise disjointness.

Therefore $W_t\subset N_t(F)\setminus F$. Since $F$ is compact, $N_{1/n}(F)\downarrow F (n\to\infty)$. The sets $N_{1/n}(F)$ have finite measure, so continuity of Lebesgue measure from above gives $\mathcal L^2\bigl(N_{1/n}(F)\setminus F\bigr) \longrightarrow0$. Choose $n$ sufficiently large that $t:=\frac1n<\frac L4$ and $\mathcal L^2\bigl(N_t(F)\setminus F\bigr) < \frac{a}{128}$. Then $\frac{a}{128} \le \mathcal L^{2,*}(W_t) \le \mathcal L^2\bigl(N_t(F)\setminus F\bigr) <\frac{a}{128}$, a contradiction. Therefore \(\mathcal L^2(F)=0\).

\noindent
\textbf{Step 5: Completion of the theorem.}

Suppose, toward a contradiction, that $\mathcal L^2(\mathcal{V})>0.$ We have shown that $\mathcal L^2(P)=0$ and $\mathcal L^2(T)=0.$ Set $R:=\mathcal{V} \setminus(P\cup T)$. Since $\mathcal{V}$ is Lebesgue measurable, $P$ is Borel, and $T$ is null and hence Lebesgue measurable, the set $R$ is Lebesgue measurable. Moreover, \[\mathcal L^2(R)=\mathcal L^2(\mathcal{V})>0.\] Since $K=P\cup E$, every point of $R$ belongs to $E$, and hence $R\subset E\setminus T$. By inner regularity of Lebesgue measure, there exists a compact set $F\subset R $ such that \(\mathcal L^2(F)>0\).
\par
For every $x\in F$, choose one $A_x$ satisfying $x\in\operatorname{rbd}A_x$. Choose a nondegenerate closed segment $S_x$ starting at $x$ and extending inward along $A_x$, so that $S_x\setminus\{x\}\subset A_x$. If $A_x$ is a bounded interval of length $\ell_x$, choose $\operatorname{length}(S_x)<\frac{\ell_x}{3}.$ If $A_x$ is a ray, choose any finite positive length.
\par
We claim that \(S_x\cap S_y=\varnothing (x\neq y)\). If $A_x=A_y$, then that common member must be a bounded interval whose two endpoints are $x$ and $y$. Since both rooted segments have length less than one third of the full interval length, they are disjoint. Now suppose $A_x\neq A_y$. If $z\in S_x\cap S_y$ and $z\neq x,y$, then $z\in A_x\cap A_y$, contradicting pairwise disjointness of the original family. If $z=x$, then $x\in S_y\setminus\{y\}\subset A_y$. Since $x$ is an endpoint of $A_x$, it follows that $x$ is a T-junction: $x\in T$, contrary to $F\subset R$. The case $z=y$ is identical. Therefore $S_x\cap S_y=\varnothing (x\neq y)$. Thus $F$ is a compact set of distinguished endpoints of a pairwise disjoint family of nondegenerate closed line segments. By the one-sided Klee--Martin lemma, $\mathcal L^2(F)=0$, contradicting $\mathcal L^2(F)>0$. Hence \(\mathcal L^2(K)=0\).
$\hfill \blacksquare$ 

\begin{Example}
\label{example3}
\par
We want to find an example of a convex partition preserved by a particular martingale coupling, \footnote{It means that there exists a feasible martingale transport that maps each point only to points within the same convex cell.} and this partition is strictly finer than the irreducible convex paving.\footnote{The irreducible paving is the $\{A_x\}_{x \in X}$ defined in Proposition \ref{pr1}.} If such an example exists, then the finest convex partition preserving the convex-order relationship within each cell can be strictly finer than the irreducible convex paving introduced by \citet{de2019irreducible}. The following example demonstrates that a convex partition preserved by a particular martingale coupling can indeed be strictly finer than the irreducible convex paving.
\par

\begin{figure}[htbp]\caption{Example 3}
\centering
\includegraphics[width=0.5\textwidth]{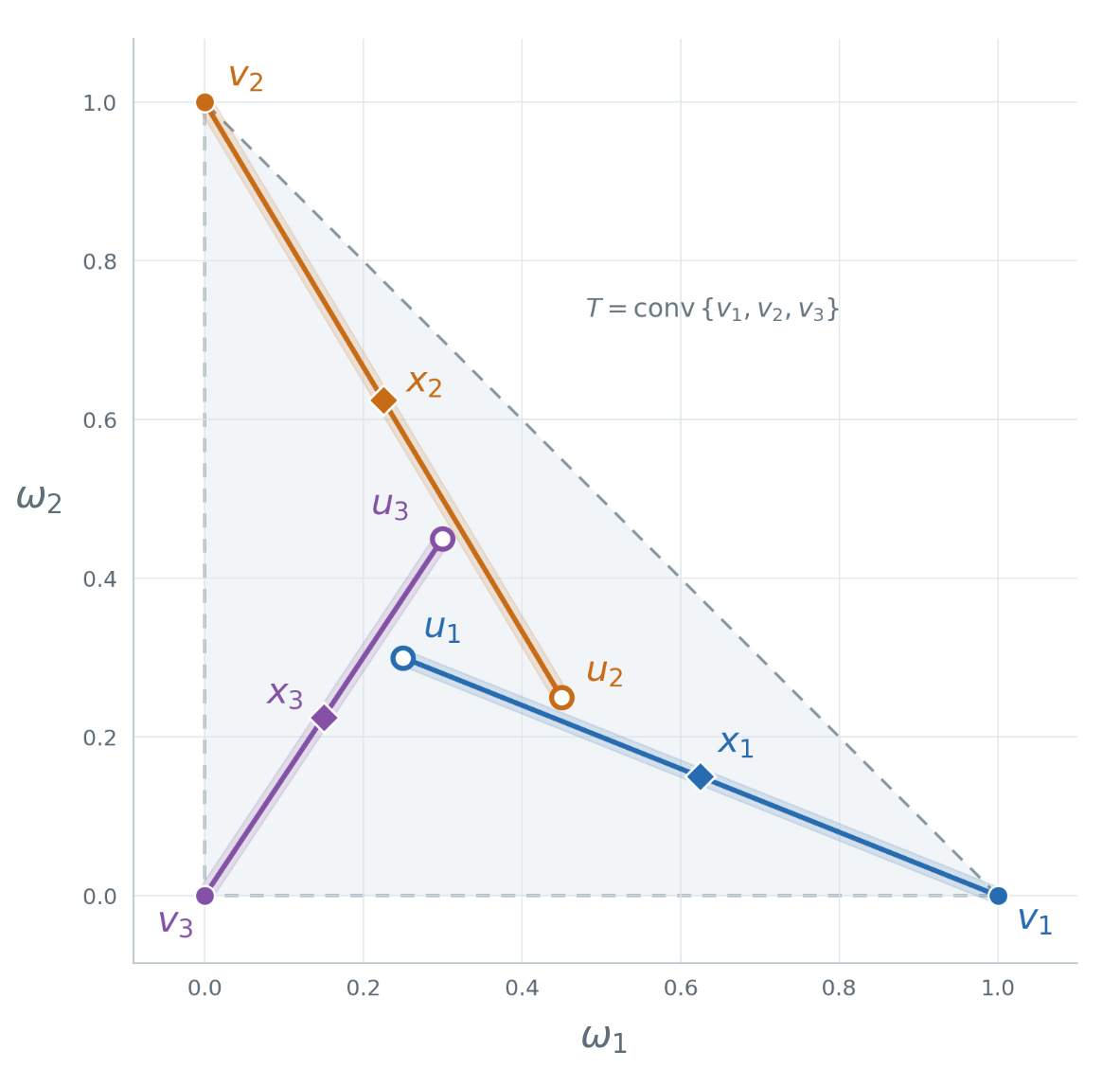}
\end{figure}

Consider the points
\[
\begin{aligned}
v_1&=(1,0),&
v_2&=(0,1),&
v_3&=(0,0),\\
u_1&=\left(\frac14,\frac3{10}\right),&
u_2&=\left(\frac9{20},\frac14\right),&
u_3&=\left(\frac3{10},\frac9{20}\right).
\end{aligned}
\]
For \(i=1,2,3\), set
\[
x_i=\frac{v_i+u_i}{2},
\qquad
S_i=[v_i,u_i],
\qquad
T=\operatorname{conv}\{v_1,v_2,v_3\}.
\]
Let \(r=1/100\), let \(B_r(z)\) denote the open Euclidean ball of radius
\(r\) centered at \(z\), and define the open convex tubes
\[
C_i=S_i+B_r(0).
\]
A direct calculation gives \(\min_{i\ne j}\operatorname{dist}(S_i,S_j)
=\frac{3}{20\sqrt{34}}>2r.\)

Hence the tubes are pairwise disjoint, and \(\mathcal C =\{C_1,C_2,C_3\} \cup \left\{\{z\}:z\in\mathbb R^2\setminus\bigcup_{i=1}^3 C_i \right\}\) is a Borel measurable partition into relatively open convex sets. Write \(\mathcal C(x)\) for the cell containing \(x\). For \(z\in\mathbb R^2\), let \(\eta_z\) be the uniform probability measure
on \(B_r(z)\). Its barycenter is \(z\). Define
\[
\mu=\frac13\sum_{i=1}^3\delta_{x_i},
\qquad
\nu=\frac16\sum_{i=1}^3
\bigl(\eta_{v_i}+\eta_{u_i}\bigr).
\]
In particular, \(\nu\) is absolutely continuous with respect to
two-dimensional Lebesgue measure.
\par
Define a coupling \(P\) by
\[
P(dx,dy)
=
\frac13\sum_{i=1}^3\delta_{x_i}(dx)P_{x_i}(dy),
\qquad
P_{x_i}=\frac12\eta_{v_i}+\frac12\eta_{u_i}.
\]
Its marginals are \(\mu\) and \(\nu\), and
\[
\int_{\mathbb R^2} y\,P_{x_i}(dy)
=\frac{v_i+u_i}{2}=x_i.
\]
Therefore \(P\in\mathcal M(\mu,\nu)\), and in particular
\(\mu\preceq_{\mathrm{cx}}\nu\). Since both
\(B_r(v_i)\) and \(B_r(u_i)\) lie in \(C_i\),
\[
P_{x_i}(C_i)=1
\quad\text{for every }i,
\qquad
P\bigl(Y\in\mathcal C(X)\bigr)=1.
\]
Thus this coupling transports each source entirely within its own cell.
\par
To identify the irreducible paving, construct another coupling \(Q\). Set
\[
Q_{x_i}
=
\frac12\eta_{u_{i+1}}
+\frac25\eta_{v_i}
+\frac1{40}\eta_{v_{i+1}}
+\frac3{40}\eta_{v_{i+2}} \quad\text{for }i=1,2,3,
\]
and 
\[Q(dx,dy)
=
\frac13\sum_{i=1}^3\delta_{x_i}(dx)Q_{x_i}(dy).\]

The coordinates above satisfy
\(\frac12u_{i+1}
+\frac25v_i
+\frac1{40}v_{i+1}
+\frac3{40}v_{i+2}
=x_i,
\qquad i=1,2,3.\) Moreover,
\(\frac13\sum_{i=1}^3Q_{x_i}
=
\frac16\sum_{i=1}^3
\bigl(\eta_{u_i}+\eta_{v_i}\bigr)
=\nu.\) Consequently \(Q\in\mathcal M(\mu,\nu)\).

Every \(Q_{x_i}\) assigns positive mass to the disks around all three
vertices. Since \(u_1,u_2,u_3\in T\), it follows that
\[
\overline{\operatorname{conv}}
\bigl(\operatorname{supp}Q_{x_i}\bigr)
=
T+\overline{B_r(0)}
=
\overline{\operatorname{conv}}
\bigl(\operatorname{supp}\nu\bigr),
\qquad i=1,2,3.
\]
For any \(R\in\mathcal M(\mu,\nu)\), each \(R_{x_i}\) is absolutely
continuous with respect to \(\nu\), because \(\mu(\{x_i\})>0\).
Thus no such conditional distribution has a larger closed convex support
than \(Q_{x_i}\). By the maximal-support characterization of the
irreducible paving in Theorem 2.1 \citet{de2019irreducible},
\[
I(x_1)=I(x_2)=I(x_3)
=
\operatorname{ri}\bigl(T+\overline{B_r(0)}\bigr)
=
T+B_r(0).
\]
In particular, \(C_i\subsetneq I(x_i),\qquad i=1,2,3.\) Every remaining cell of \(\mathcal C\) is a singleton and therefore lies in a cell of the irreducible paving. Hence \(\mathcal C\) strictly refines that paving.
\end{Example}

\noindent\textbf{Proof of Proposition \ref{pr3}.}
Put $I=\operatorname{ri}X$, $L=\operatorname{span}(X-X)$, and
$r=\dim L$. By the concentration hypothesis, $\mu(I)=1$ and the
canonical conditional closed convex support is $X$, $\mu$-almost
everywhere. In particular,
$X=\operatorname{conv}(\operatorname{supp}\mu_0)$: one inclusion
follows from the second marginal constraint and the canonical support
property, and the other from $\mu_0(X)=1$ and compactness of $X$.
The case $r=0$ is immediate.

\medskip
\noindent\textbf{Step 1: An affine-annihilating perturbation.}
Suppose that $\operatorname{supp}\mu$ is not affinely independent.
We construct a nonzero bounded Borel function $h$, supported on a compact
set $K\subset I$, such that
\begin{equation}\label{eq:zero-mass-moment}
\int_X h\,d\mu=0,
\qquad
\int_X xh(x)\,d\mu(x)=0.
\end{equation}
If $\operatorname{supp}\mu=\{x_1,\ldots,x_m\}$ is finite, all its
points belong to $I$. Choose coefficients $b_i$, not all zero, with
$\sum_i b_i=0$ and $\sum_i b_ix_i=0$, and set
$h(x_i)=b_i/\mu(\{x_i\})$, with $h=0$ elsewhere and
$K=\{x_1,\ldots,x_m\}$.

If the support is infinite, choose $r+2$ distinct points in
$I\cap\operatorname{supp}\mu$ and pairwise disjoint relatively open
neighborhoods $B_i$ whose closures are compact subsets of $I$.
Such points exist because $\mu(I)=1$. Each $\mu(B_i)>0$.
The $r+2$ vectors
\[
v_i=\left(\mu(B_i),\int_{B_i}x\,d\mu(x)\right)
\]
lie in the $(r+1)$-dimensional space
$\operatorname{span}\{(1,x):x\in X\}$ and are therefore linearly
dependent. Choose coefficients $a_i$, not all zero, with
$\sum_i a_iv_i=0$. Taking
$h=\sum_i a_i\mathbf1_{B_i}$ and $K=\bigcup_i\overline{B_i}$ gives
\eqref{eq:zero-mass-moment}. In either case, put
\[
\sigma=h\mu,
\qquad
H=\int_X|h|\,d\mu>0.
\]
Then $\sigma$ is nonzero, is supported on $K\subset I$, and annihilates
every affine function.

\medskip
\noindent\textbf{Step 2: Robust anchor neighborhoods.}
There are finitely many points $y_1,\ldots,y_N$ in
$S=\operatorname{supp}\mu_0$ such that
\[
P_0=\operatorname{conv}\{y_1,\ldots,y_N\},
\qquad
\operatorname{aff}P_0=\operatorname{aff}X,
\qquad
K\subset\operatorname{ri}P_0.
\]
Indeed, for each $x\in K$, take an $r$-dimensional simplex
$\Delta_x\subset I$ containing $x$ in its relative interior.
Its vertices are finite convex combinations of points of $S$, since
$X=\operatorname{conv}S$. A finite subcover of $K$ by the sets
$\operatorname{ri}\Delta_x$ yields the required finite set of anchors.
Choose $\delta>0$ with $K+\delta B_L\subset P_0$, where
$B_L=\{v\in L:\|v\|\leq1\}$, and set
\[
U_j=\{y\in X:\|y-y_j\|<\delta/2\}.
\]
Every $U_j$ has positive $\mu_0$-measure. If $z_j\in U_j$ for all $j$,
the support functions satisfy, for every unit vector $u\in L$,
\[
\max_j u\cdot z_j
\geq\max_j u\cdot y_j-\delta/2
\geq\sup_{z\in K}u\cdot z+\delta/2.
\]
Consequently,
\begin{equation}\label{eq:robust-convex-hull}
K\subseteq\operatorname{conv}\{z_1,\ldots,z_N\}
\quad\text{whenever }z_j\in U_j\text{ for all }j.
\end{equation}

\medskip
\noindent\textbf{Step 3: Uniformly positive anchor probabilities.}
By Corollary \ref{Co2}, there is $P\in\mathcal M(\mu,\mu_0)$ whose
conditional distributions have support $S$, $\mu$-almost everywhere.
In particular, $P_x(U_j)>0$ for all $j$, simultaneously for almost
every $x$. For a Borel version of this kernel define
\[
E_m=\{x\in I:P_x(U_j)\geq1/m\text{ for all }j\}.
\]
Since $\mu(\bigcup_m E_m)=1$, some $E=E_{m_0}$ has positive
$\mu$-measure. Write $c=1/m_0$, so that
\begin{equation}\label{eq:uniform-positive}
P_x(U_j)\geq c
\quad\text{for all }j\text{ and }\mu\text{-almost every }x\in E.
\end{equation}

\medskip
\noindent\textbf{Step 4: A uniform bound by the convex-order slack.}
Fix a continuous convex function $\varphi:X\to\mathbb R$ and let
\[
g_\varphi(x)=\int_X\varphi(y)P_x(dy)-\varphi(x).
\]
Jensen's inequality and the marginal constraints give
$g_\varphi\geq0$ almost everywhere and
\begin{equation}\label{eq:Jensen-gap}
\int_X g_\varphi\,d\mu
=\int_X\varphi\,d\mu_0-\int_X\varphi\,d\mu.
\end{equation}
For $x\in E\subset I$ outside the exceptional null set, choose a
relative subgradient $\xi_x\in L$ and define
\[
R_x(z)=\varphi(z)-\varphi(x)-\xi_x\cdot(z-x)\geq0,
\qquad z\in X.
\]
The subgradient exists because $x\in\operatorname{ri}X$; its supporting
inequality extends to the boundary by continuity. The martingale
property gives $\int_XR_x\,dP_x=g_\varphi(x)$. Thus, for each
$\eta>0$ and each $j$, there is $z_j\in U_j$ such that
\[
R_x(z_j)\leq\frac{g_\varphi(x)}{c}+\eta,
\]
because the average of $R_x$ under the normalized restriction of
$P_x$ to $U_j$ is at most $g_\varphi(x)/c$.
By \eqref{eq:robust-convex-hull} and convexity of $R_x$,
\[
\sup_{z\in K}R_x(z)
\leq\max_jR_x(z_j)
\leq\frac{g_\varphi(x)}c+\eta.
\]
Letting $\eta\downarrow0$ and using the affine-annihilation property
of $\sigma$ yields
\[
\left|\int_X\varphi\,d\sigma\right|
=\left|\int_XR_x\,d\sigma\right|
\leq H\sup_{z\in K}R_x(z)
\leq\frac Hc\,g_\varphi(x).
\]
Integrating over $E$ and using \eqref{eq:Jensen-gap}, we obtain
\begin{equation}\label{eq:uniform-slack}
\int_X\varphi\,d\mu_0-\int_X\varphi\,d\mu
\geq\kappa\left|\int_X\varphi\,d\sigma\right|,
\qquad
\kappa=\frac{c\mu(E)}H>0.
\end{equation}
The constant is independent of $\varphi$. No measurable choice of
$\xi_x$ is needed: only the measurable function $g_\varphi$ is
integrated.

\medskip
\noindent\textbf{Step 5: Two distinct feasible perturbations.}
Choose $0<\varepsilon<\min\{\kappa,\|h\|_{L^\infty(\mu)}^{-1}\}$
and define
$\mu^\pm=(1\pm\varepsilon h)\mu$. They are probability measures with
the same barycenter as $\mu$. For every continuous convex $\varphi$,
\[
\int_X\varphi\,d\mu_0-\int_X\varphi\,d\mu^\pm
\geq(\kappa-\varepsilon)
\left|\int_X\varphi\,d\sigma\right|\geq0.
\]
In particular this inequality holds for restrictions to $X$ of all
finite convex functions on $\mathbb R^n$, so
$\mu^\pm\in MPC(\mu_0)$. Since $\sigma\ne0$,
$\mu^+\ne\mu^-$ and $\mu=(\mu^++\mu^-)/2$. Thus $\mu$ is not extreme.

$\hfill\blacksquare$

\par
Given two probability measure $\mu_1,\mu_2 \in MPC(\mu_0)$, we say $\mu_1$ has finer concentration regions than $\mu_2$ with respect to $\mu_0$ if for the corresponding $\{A^1_x\}_{x \in X}$ and $\{A^2_x\}_{x \in X}$ defined in Proposition \ref{pr1}, we have $A^1_x \subset A^2_x$ for $\mu_1$-a.s. all $x$.
\begin{Lemma}\label{lm9}
Given two probability measures $\mu_1, \mu_2 \in MPC(\mu_0)$, if $\mu=\alpha\mu_1+(1-\alpha)\mu_2, \alpha \in (0,1)$, then both $\mu_1$ and $\mu_2$ have finer concentration regions than $\mu$ with respect to $\mu_0$.
\end{Lemma}

\noindent\textbf{Proof of Lemma \ref{lm9}.}
Let $\hat p^i\in\mathcal M(\mu_i,\mu_0)$, $i=1,2$, be the
couplings that induce $\{A_x^i\}_{x\in X}$ as in
Proposition \ref{pr1}. Since $\mu=\alpha\mu_1+(1-\alpha)\mu_2$,
the mixture $q:=\alpha\hat p^1+(1-\alpha)\hat p^2$ belongs to
$\mathcal M(\mu,\mu_0)$. Its conditional kernel is
\[
q_x
=
\alpha\frac{d\mu_1}{d\mu}(x)\hat p_x^1
+
(1-\alpha)\frac{d\mu_2}{d\mu}(x)\hat p_x^2
\qquad\text{for $\mu$-almost every $x$.}
\]
The coefficient of $\hat p_x^i$ is positive for
$\mu_i$-almost every $x$. Hence, by Proposition \ref{pr1},
the maximal coupling $\hat p\in\mathcal M(\mu,\mu_0)$ satisfies
\[
C_x^i:=\hat{supp}(\hat p_x^i)
\subseteq \hat{supp}(q_x)
\subseteq C_x:=\hat{supp}(\hat p_x)
\qquad\text{for $\mu_i$-almost every $x$, $i=1,2$.}
\]
By Proposition \ref{pr1}, $x\in\operatorname{ri}C_x^i
\cap\operatorname{ri}C_x$. For any
$y\in\operatorname{ri}C_x^i$ with $y\ne x$, the ray from $x$
through $y$ extends beyond $y$ to some $z\in C_x^i\subseteq C_x$.
Since $x\in\operatorname{ri}C_x$, the segment $[x,z)$ lies in
$\operatorname{ri}C_x$. Thus $y\in\operatorname{ri}C_x$, and
\(A_x^i=\operatorname{ri}C_x^i
\subseteq \operatorname{ri}C_x=A_x \text{ for $\mu_i$-almost every $x$, $i=1,2$.}\)

\hfill$\blacksquare$

\end{document}